\documentclass[preprint,3p,number,sort&compress]{elsarticle}

\usepackage[utf8]{inputenc}
\usepackage[T1]{fontenc}
\usepackage{amsmath,amssymb,url,paralist}

\providecommand*{\nat}[0]{\ensuremath{\mathbb{N}}}
\providecommand*{\seq}[3]{\ensuremath{#1_{#2}, \dotsc, #1_{#3}}}

\providecommand*{\abs}[1]{\ensuremath{\lvert #1 \rvert}}

\allowdisplaybreaks

\DeclareMathOperator{\pos}{pos}
\DeclareMathOperator{\Reg}{Reg}
\DeclareMathOperator{\var}{var}
\DeclareMathOperator{\lca}{lcp}
\DeclareMathOperator{\height}{ht}
\DeclareMathOperator{\SSP}{SP}
\DeclareMathOperator{\SP}{P}

\DeclareMathOperator{\dep}{dep}
\DeclareMathOperator{\link}{links}
\DeclareMathOperator{\rk}{rk}
\DeclareMathOperator{\id}{id}
\DeclareMathOperator{\used}{used}
\DeclareMathOperator{\branch}{br}

\newtheorem{theorem}{Theorem}
\newtheorem{lemma}[theorem]{Lemma}

\newtheorem{corollary}[theorem]{Corollary}
\newtheorem{example}[theorem]{Example}
\newdefinition{definition}[theorem]{Definition}
\newproof{proof}{Proof}

\title{Composition Closure of 
  Linear Extended Top-down Tree Transducers}

\author[szeged]{Zolt\'an F\"ul\"op\fnref{zf}}
\ead{fulop@inf.u-szeged.hu}

\author[stuttgart]{Andreas Maletti\corref{cor}\fnref{am}}
\ead{maletti@ims.uni-stuttgart.de}

\address[szeged]{Department of Foundations of Computer Science,
  University of Szeged \\ \'Arp\'ad t\'er~2, H-6720 Szeged, Hungary}

\address[stuttgart]{Universit\"at Stuttgart, Institut f\"ur
  Maschinelle Sprachverarbeitung \\ Pfaffenwaldring~5b, 70569
  Stuttgart, Germany}

\cortext[cor]{Corresponding author}
\fntext[zf]{Supported by the program
  T\'AMOP-4.2.1/B-09/1/KONV-2010-0005 of the Hungarian National
  Development Agency.}
\fntext[am]{Supported by the German Research Foundation~(DFG) grant
  MA\,/\,4959\,/\,1-1.}

\journal{arXiv}

\begin{document}

\begin{abstract}
  Linear extended top-down tree transducers (or synchronous
  tree-substitution grammars) are popular formal models of tree
  transformations.  The expressive power of compositions of such
  transducers with and without regular look-ahead is investigated.
  In particular, the restrictions of nondeletion,
  $\varepsilon$-freeness, and strictness are considered.
  The composition hierarchy turns out to be finite for all
  $\varepsilon$-free (all rules consume input) variants of these
  transducers except for nondeleting $\varepsilon$-free linear
  extended top-down tree transducers.  The least number of transducers
  needed for the full expressive power of arbitrary compositions is
  presented.  In all remaining cases (incl.\ nondeleting
  $\varepsilon$-free linear extended top-down tree transducers) the
  composition hierarchy does not collapse.
\end{abstract}


\maketitle

\section{Introduction}
\label{sec:Intro} 
The top-down tree transducer is a simple formal model that encodes a
tree transformation (i.e., a relation on trees).  It was introduced
in~\cite{rou70,tha70} and intensively studied thereafter
(see~\cite{gecste84,gecste97,fulvog98} for an overview).  Roughly
speaking, a top-down tree transducer processes the input tree
symbol-by-symbol and specifies in its rules, how to translate an input
symbol into an output tree fragment together with instructions on how
to process the subtrees of the input symbol.  This asymmetry between
input (single symbol) and output (tree fragment) was removed in
extended top-down tree transducers~(xt), which were introduced and
studied in~\cite{arndau75,arndau76}.  In an xt the left-hand side of a
rule now contains an input tree fragment, in which each variable can
occur at most once as a placeholder for a subtree.  In particular, the
input tree fragment can even be just a variable, which matches every
tree, and such rules are called $\varepsilon$-rules.  In this
contribution we consider linear xt~(l-xt), in which the right-hand
side of each rule contains each variable at most once as well.
Restricted variants of~l-xt are used in most approaches to
syntax-based machine translation~\cite{graknimay08,knigra05}.

We also add regular look-ahead~\cite{eng77} (i.e., the ability to
check a regular property for the subtrees in an input tree fragment)
to l-xt, so our most expressive model is the linear extended top-down
tree transducer with regular look-ahead~(l-xt${}^{\text R}$).
Contrary to most of the literature~\cite{eng77,malgrahopkni07} we
present our model as a synchronized grammar~\cite{chi06} because we
sometimes use the auxiliary link structure in our proofs.  Instead of
variables in the left-hand side and a state-variable combination in
the right-hand side of a rule, we immediately only use states with the
restriction that each state can occur at most once in the left-hand
side and at most once in the right-hand side.  Moreover, all states
that occur in the right-hand side must also occur in the left-hand
side.  In this way, for each rule the states establish implicit links
(a state links its occurrence in the left-hand side with its
occurrence in the right-hand side), which form a bijection between a
subset of the state occurrences in the left-hand side and all state
occurrences in the right-hand side.  The state occurrences (in the
left-hand side) that do not participate in the bijection (i.e., those
states that exclusively occur in the left-hand side) can restrict the
acceptable subtrees at their position with the help of regular
look-ahead~\cite{eng77}.  The implicit links in a rule are made
explicit in a derivation, and a rule application expands (explicitly)
linked state occurrences at the same time.  Example \ref{ex:lXT} shows
an l-xt${}^{\text R}$, for which we illustrate a few derivation steps
in Figure~\ref{fig:xDeriv}.  The tree transformation computed by the
example l-xt${}^{\text R}$ is shown in Example~\ref{ex:Sem2}.  In the
following, we use l-XT${}^{\text R}$~and~l-XT to denote the class of
all tree transformations computed by l-xt${}^{\text R}$~and~l-xt,
respectively.

The expressive power of the various subclasses of l-XT${}^{\text R}$
is already well understood~\cite{malgrahopkni07,fulmalvog11}.
However, in practice complex systems are often specified with the help
of compositions of tree transformations~\cite{mayknivog10} because it
is much easier to develop (or train) small components that manage a
part of the overall transformation.  Consequently,
\cite{knigra05}~and~others declare that closure under composition is a
very desirable property for classes of tree transformations
(especially in the area of natural language processing).  If a
class~${\cal C}$ of tree transformations is closed under composition,
then any composition chain $\tau_1 \mathbin; \dotsm \mathbin; \tau_n$
of tree transformations~$\seq \tau1n$ of~${\cal C}$ can be replaced by
a single tree transformation $\tau \in {\cal C}$.  If ${\cal
  C}$~represents the class of all tree transformations computable by a
device, then closure under composition means that we can replace any
composition chain specified by several devices by just a single
device, which enables an efficient modular development.
Unfortunately, neither l-XT${}^{\text R}$ nor l-XT are closed under
composition~\cite{arndau76,arndau82,malgrahopkni07}.

In general, for a class~$\cal C$ of tree transformations (that
contains the identity transformation) we obtain a composition
hierarchy ${\cal C} \subseteq {\cal C}^2 \subseteq {\cal C}^3
\subseteq \dotsb {}$, where ${\cal C}^n$~denotes the $n$-fold
composition of~${\cal C}$.  The class ${\cal C}$ might be closed under
composition at power $n$ (i.e., ${\cal C}^n = {\cal C}^{n+1}$) or its
composition hierarchy might be infinite (i.e., ${\cal C}^n \subsetneq
{\cal C}^{n+1}$ for all $n$).  In the former case, we say that the
composition hierarchy of ${\cal C}$ collapses at power~$n$, which also
yields that ${\cal C}^n = {\cal C}^{m}$ for all $m \geq n$.  In
particular, ${\cal C}$ is closed under composition if its composition
hierarchy collapses at power~$1$.  We note that in practice (e.g., in
machine translation) the classes that are closed under composition at
a small finite power are also important because for such classes we
can limit the length of composition chains~\cite{mayknivog10}.  In
this contribution, we investigate the composition hierarchy of the
classes \mbox{l-XT${}^{\text R}$}~and~l-XT together with their subclasses
determined by the properties: $\varepsilon$-freeness, strictness, and
nondeletion, which are abbreviated by~`$\not\!\varepsilon$', `s', and
`n', respectively.  Roughly speaking, $\varepsilon$-freeness yields
that all rules are $\varepsilon$-free, strictness guarantees that the
right-hand side of each rule contains an output symbol, and
nondeletion requires that for each rule exactly the same states occur
in the left- and right-hand side.  We use the property abbreviations
in front of l-XT${}^{\text R}$~and~l-XT to obtain the class of all
tree transformations computable by such restricted l-xt${}^{\text R}$
and l-xt, respectively.  For instance,
$\text{$\mathord{\not{\!\varepsilon}}$sl-XT}^{\text R}$ denotes the
class of all tree transformations computed by $\varepsilon$-free and
strict l-xt${}^{\text R}$.

It is known that none of our considered classes is closed under
composition~\cite[Section~3.4]{arndau82}.  In addition, it is known
that $\text{$\mathord{\not{\!\varepsilon}}$snl-XT} =
\text{$\mathord{\not{\!\varepsilon}}$snl-XT}^{\text R}$ is closed at
power~2~\cite[Section~II-2-2-3-3]{dau77}.  We complete the picture as
follows.  For each of the remaining classes, we either provide the
least power at which the class is closed under composition or show
that the composition hierarchy of the class is infinite (denoted
by~$\infty$).  Our results (together with the mentioned existing
result) are presented in Table~\ref{tab:results}.

\begin{table}[t]
  \begin{center}
    \begin{tabular}{|l|c|c|}
      \hline
      & & \\[-1.5ex]
      \; Class & \mbox{\;} Least power of closedness\mbox{\;} & Proved in
      \\[1ex] \hline
      & & \\[-1.5ex]
      \; $\text{$\mathord{\not{\!\varepsilon}}$snl-XT} =
      \text{$\mathord{\not{\!\varepsilon}}$snl-XT}^{\text R}$ \; & 2 & 
      \; \cite[Section~II-2-2-3-3]{dau77} \; \\[1ex] \hline
      & & \\[-1.5ex]
      \; $\text{$\mathord{\not{\!\varepsilon}}$sl-XT}^{\text R}$,
      $\text{$\mathord{\not{\!\varepsilon}}$sl-XT}$ & $2$ &
      Theorem~\ref{thm:min1} \\[1ex]
      \; $\text{$\mathord{\not{\!\varepsilon}}$l-XT}^{\text R}$ & $3$ &
      Theorem \ref{thm:min2} \\[1ex]
      \; $\text{$\mathord{\not{\!\varepsilon}}$l-XT}$ & $4$ &
      Corollary \ref{cor:min3} \\[1ex] 
      \hline
      & & \\[-1.5ex]
      \; otherwise & $\infty$ & Theorem~\ref{cor:eps} \\[1ex]
      \hline
    \end{tabular}
  \end{center}
  \caption{Characterization of the composition hierarchies.}
  \label{tab:results}
\end{table}

Our contribution is organized as follows.  Section~\ref{sec:prelim}
recalls the necessary concepts and introduces our notation.  We
continue in Section~\ref{sec:xtop} with the formal introduction of our
main model (l-xt${}^{\text R}$) including its syntax and semantics and
the restrictions that we consider later.  In addition, we recall some
known equalities between certain fundamental classes of tree
transformations in preparation for our first main results.  In
Section~\ref{sec:upper} we give a power at which the classes
$\text{$\mathord{\not{\!\varepsilon}}$sl-XT}$,
$\text{$\mathord{\not{\!\varepsilon}}$sl-XT}^{\text R}$,
$\text{$\mathord{\not{\!\varepsilon}}$l-XT}$, and
$\text{$\mathord{\not{\!\varepsilon}}$l-XT}^{\text R}$ of tree
transformations are closed under composition
(cf. Table~\ref{tab:results}).  This is completed in
Section~\ref{sec:lower}, where we conclude that the presented powers
(Table~\ref{tab:results}) are minimal.  Finally, in
Section~\ref{sec:complete} we prove that the composition hierarchy of
the remaining classes is infinite.

\section{Notation}
\label{sec:prelim}
We denote the set of all nonnegative integers by~$\nat$.  The set of
all \emph{finite words} (finite sequences) over~a set~$S$ is~$S^* =
\bigcup_{n \in \nat} S^n$, where $S^0 = \{\varepsilon\}$~contains only
the \emph{empty word}~$\varepsilon$.  The \emph{length} of a word~$w
\in S^*$ is the unique~$n \in \nat$ such that $w \in S^n$.  We write
$\abs w$ for the length of~$w$.  The \emph{concatenation} of two
words~$v, w \in S^*$ is denoted by $v.w$~or simply~$vw$.

Every subset of~$S \times T$ is a \emph{relation} from~$S$ to~$T$.
Given relations
$R_1 \subseteq S \times T$ and $R_2 \subseteq T \times U$, the
 \emph{inverse} of $R_1$ is the
relation~$R_1^{-1} = \{ (t, s) \mid (s, t) \in R_1\}$, and the
\emph{composition} of $R_1$ and $R_2$ is the
relation \[ R_1 \mathbin; R_2 = \{ (s, u) \mid \exists t \in T\colon
(s, t) \in R_1, (t, u) \in R_2\} \enspace. \] These notions and
notations are lifted to classes ${\cal C}_1$~and~${\cal C}_2$ of
relations in the usual manner.  Namely, we let ${\cal C}_1^{-1} = \{
R_1^{-1} \mid R_1 \in {\cal C}_1\}$ and
\[ {\cal C}_1 \mathbin; {\cal C}_2 = \{R_1 \mathbin; R_2 \mid R_1 \in
{\cal C}_1, R_2 \in {\cal C}_2\} \enspace. \] Moreover, the
\emph{powers} of a class~${\cal C}$ are defined by ${\cal C}^1 = {\cal
  C}$ and ${\cal C}^{n+1} ={\cal C}^n \mathbin; {\cal C}$ for~$n
\geq 1$.  The \emph{composition hierarchy} (resp.\ \emph{composition
  closure}) \emph{of~${\cal C}$} is the family $({\cal C}^n \mid n
\geq 1)$ (resp.\ the class $\bigcup_{n \geq 1} {\cal C}^n$).  If
${\cal C}^{n+1} = {\cal C}^n$, then \emph{${\cal C}$~is closed under
  composition at power~$n$}.  For $n = 1$ we shorten this to just
\emph{${\cal C}$~is closed under composition}.  If ${\cal C}$ is
closed under composition at power~$n$, then $\bigcup_{1 \leq i \leq n}
{\cal C}^i$ is the composition closure of~${\cal C}$.  Moreover, we
note that if ${\cal C}$~contains the identity relations, then ${\cal
  C}^n \subseteq {\cal C}^{n+1}$ for all~$n \geq 1$, so that in this case ${\cal
  C}^n$ is the composition closure of~${\cal C}$ provided that ${\cal C}$
 is closed under composition at
power~$n$.  Our classes~${\cal C}$ of tree transformations will always
contain the identity relations.

An \emph{alphabet}~$\Sigma$ is a nonempty and finite set, of which the
elements are called \emph{symbols}.  The alphabet~$\Sigma$ is
\emph{ranked} if there additionally is a mapping~$\mathord{\rk} \colon
\Sigma \to \nat$ that assigns a rank to each symbol.  We let $\Sigma_k
= \{\sigma \in \Sigma \mid \rk(\sigma) = k\}$ for every $k \in \nat$.
Often the mapping~`$\mathord{\rk}$' is obvious from the context, so we
typically denote ranked alphabets by~$\Sigma$ alone.  If it is not
obvious, then we use the notation~$\sigma^{(k)}$ to indicate that the
symbol~$\sigma$ has rank~$k$.  For the rest of this paper, $\Sigma$,
$\Delta$, and~$\Gamma$ will denote arbitrary ranked alphabets if not
specified otherwise.

For every set~$T$, let 
\[ \Sigma(T) = \{\sigma(\seq t1k) \mid \sigma \in \Sigma_k, \seq t1k
\in T\} \enspace. \] Let $S$~be a set with $S \cap \Sigma =
\emptyset$.  The set~$T_\Sigma(S)$ of \emph{$\Sigma$-trees with leaf
  labels~$S$} is the smallest set~$U$ such that $S \subseteq U$ and
$\Sigma(U) \subseteq U$.  We write~$T_\Sigma$
for~$T_\Sigma(\emptyset)$, and any subset of~$T_\Sigma$ is a
\emph{tree language}.  The \emph{height}~$\height(t)$ of a tree~$t \in
T_\Sigma(S)$ is defined such that $\height(s) = 0$ for all $s \in S$
and \[ \height(\sigma(\seq t1k)) = 1 + \max\ \{\height(t_i) \mid 1
\leq i \leq k\} \] for all $\sigma\in\Sigma_k$ and $\seq t1k \in
T_\Sigma(S)$.  Given a unary symbol~$\gamma \in \Sigma_1$ and a
tree~$t \in T_\Sigma(S)$, we write~$\gamma^k(t)$ for the
tree~$\underbrace{\gamma(\dotsm \gamma}_{k \text{ times}}(t) \dotsm)$. 

The set~$\pos(t) \subseteq \nat^*$ of \emph{positions} of~$t \in
T_\Sigma(S)$ is inductively defined by $\pos(s) = \{\varepsilon\}$ for
every $s \in S$ and
\[ \pos(\sigma(\seq t1k)) = \{\varepsilon\} \cup \bigcup_{i = 1}^k \{
iw \mid w \in \pos(t_i)\} \] for every $\sigma \in \Sigma_k$ and $\seq
t1k \in T_\Sigma(S)$.  For words $v,w \in \nat^*$, we denote the
longest common prefix of $v$~and~$w$ by~$\lca(v,w)$.  Note that
$\lca(v, w) \in \pos(t)$ for all $v, w \in \pos(t)$ because
$\pos(t)$~is prefix-closed.  The positions~$\pos(t)$ are partially
ordered by the prefix order~$\preceq$ on~$\nat^*$ [i.e., $v \preceq w$
if and only if $v = \lca(v, w)$].  The size~$\abs t$ of the tree~$t
\in T_\Sigma(S)$ is~$\abs{\pos(t)}$; i.e., the number of its
positions.  Let $t \in T_\Sigma(S)$ and $w \in \pos(t)$.  The
\emph{label} of~$t$ at~$w$ is~$t(w)$, and the \emph{$w$-rooted
  subtree} of~$t$ is~$t|_w$.  Formally, $s(\varepsilon) =
s|_{\varepsilon} = s$ for every $s \in S$ and
\begin{align*}
  t(w) = \begin{cases} \sigma & \text{if } w = \varepsilon \\ t_i(v) &
    \text{if } w = iv \text{ and } i \in \nat \end{cases} \qquad
  \text{and} \qquad t|_w = \begin{cases} t & \text{if } w = \varepsilon
    \\ t_i|_v & \text{if } w = iv \text{ and } i \in \nat \end{cases}
\end{align*}
where $t = \sigma(\seq t1k)$ with $\sigma \in \Sigma_k$ and $\seq t1k
\in T_\Sigma(S)$.  For every selection~$U \subseteq S$ of leaf
symbols, we let $\pos_U(t) = \{ w \in \pos(t) \mid t(w) \in U\}$ and
$\pos_s(t) = \pos_{\{s\}}(t)$ for every $s \in S$.  The tree~$t$ is
\emph{linear} (resp.\ \emph{nondeleting}) in~$U$ if $\abs{\pos_u(t)}
\leq 1$ (resp.\ $\abs{\pos_u(t)} \geq 1$) for every $u \in U$.
Moreover, $\var(t) = \{ s \in S \mid \pos_s(t) \neq \emptyset\}$.  The
expression~$t[u]_w$ denotes the tree that is obtained from~$t \in
T_\Sigma(S)$ by replacing the subtree~$t|_w$ at~$w$ by~$u \in
T_\Sigma(S)$.

Let $U \subseteq S$ be finite, $t \in T_\Sigma(S)$, and $\theta \colon
U \to \{ L \mid L \subseteq T_\Sigma(S)\}$.  We define the tree
language~$t\theta$ by induction as follows.
\begin{compactitem}
\item $s\theta = s$ for all $s \in S \setminus U$,
\item $s\theta = \theta(s)$ for all $s \in U$, and
\item for all $\sigma \in \Sigma_k$ and $\seq t1k \in T_\Sigma(S)$
  \begin{align*}
    \sigma(\seq t1k)\theta &= \{\sigma(\seq u1k) \mid u_1 \in
    t_1\theta, \dotsc, u_k \in t_k\theta \} \enspace.
  \end{align*}
\end{compactitem}
For every $n \in \nat$ we fix the set $X_n = \{\seq x1n\}$ of
variables, which we assume to be disjoint from all ranked alphabets
considered in the paper.  Given $t \in T_\Sigma(S)$ and $\tau \colon
X_n \to T_\Sigma(S)$, we write $t[\tau(x_1), \dotsc, \tau(x_n)]$
for~$t\theta$, where $\theta(x_i) = \{\tau(x_i)\}$ for all $1 \leq i
\leq n$.

A \emph{tree homomorphism from}~$\Sigma$ \emph{to}~$\Delta$ is a
family of mappings $(h_k \mid k \in \nat)$ such that $h_k \colon
\Sigma_k \to T_\Delta(X_k)$ for every $k \in \nat$.  Such a tree
homomorphism is
\begin{compactitem}
\item \emph{linear} if for every $\sigma \in \Sigma_k$ the
  tree~$h_k(\sigma)$ is linear in~$X_k$, 
\item \emph{complete} if for every $\sigma \in \Sigma_k$ the
  tree~$h_k(\sigma)$ is nondeleting in~$X_k$,
\item \emph{strict} if $h_k \colon \Sigma_k \to \Delta(T_\Delta(X_k))$
  for every $k \in \nat$, and
\item \emph{delabeling} if $h_k \colon \Sigma_k \to X_k \cup \Delta(X_k)$
  for every $k \in \nat$.
\end{compactitem}
We abbreviate the above restrictions by `l', `c', `s', and `d'.  The
tree homomorphism~$(h_k \mid k \in \nat)$ induces a mapping $h \colon
T_\Sigma(S) \to T_\Delta(S)$ defined inductively by $h(s) = s$ for all
$s \in S$ and
\[ h(\sigma(\seq t1k)) = h_k(\sigma)[h(t_1), \dotsc, h(t_k)] \] for
all $\sigma \in \Sigma_k$ and $\seq t1k \in T_\Sigma(S)$.  As usual,
we also call the induced mapping~$h$ a tree homomorphism.  We denote
by~H the class of all tree homomorphisms, and for any combination~$w$
of~`l', `c', `s', and~`d' we denote by~$w$-H the class of all $w$-tree
homomorphisms.  For instance, lcs-H~is the class of all linear,
complete and strict tree homomorphisms.

In the following, we need \emph{regular tree
  languages}~\cite{gecste84,gecste97} and basic results for tree
automata, which recognize the regular tree languages.  The
set~$\Reg(\Gamma)$ contains all regular tree languages $L \subseteq
T_\Gamma$ over the ranked alphabet~$\Gamma$.  A particular well-known
result, which we use, states that $t\theta \in \Reg(\Gamma)$ for all
$t \in T_\Gamma(U)$ and $\theta \colon U \to \Reg(\Gamma)$.  A
\emph{bimorphism} is a triple~$B = (\psi, L, \varphi)$ consisting of a
regular tree language~$L \in \Reg(\Gamma)$, an input tree homomorphism
$\psi \colon T_\Gamma \to T_\Sigma$, and an output tree homomorphism
$\varphi \colon T_\Gamma \to T_\Delta$.  The tree transformation $B
\subseteq T_\Sigma \times T_\Delta$ computed by the bimorphism~$B$ is
the relation $B = \{ \langle \psi(t), \varphi(t) \rangle \mid t \in
L\}$.  Given two combinations $v$~and~$w$ of restrictions for tree
homomorphisms, we let $\text B(v, w)$ denote the class of all tree
transformations computed by bimorphisms $B = (\psi, L, \varphi)$ such
that $\psi$~and~$\varphi$ are tree homomorphisms of type $v$ and $w$,
respectively.

\section{Linear extended top-down tree transducers}
\label{sec:xtop}
Our main model is the linear extended top-down tree
transducer~\cite{arndau75,arndau76,knigra05,graknimay08} with regular
look-ahead (l-xt${}^{\text R}$), which is based on the classical
linear top-down tree transducer without~\cite{rou70,tha70} and with
regular look-ahead~\cite{eng77}.  We will present it in a form that is
closer to synchronized grammars~\cite{chi06} because we will use an
auxiliary structure, called the links, in a later proof.  In
synchronous grammars, equal states in the left and right-hand side of
a rule are implicitly linked, and equal states are explicitly linked
in a derivation, in which such linked states will be replaced at the
same time by a rule for that state.  In a rule of an l-xt${}^{\text
  R}$, these implicit links form a bijection between a subset of the
states in the left-hand side and all states in the right-hand side.
Thus, some states might exclusively occur in the left-hand side.  For
those states we can use the regular look-ahead to restrict the
subtrees that are acceptable at these positions.

\begin{definition}[{\protect{see~\cite[Section~2.2]{malgrahopkni07}}}]
  \label{df:xtop}
  \upshape A \emph{linear extended top-down tree transducer} with
  regular look-ahead (l-xt${}^{\text R}$) is a tuple $M=(Q, \Sigma,
  \Delta, I, R, c)$, where
  \begin{compactitem}
  \item $Q$~is a finite set of \emph{states},
  \item $\Sigma$~and~$\Delta$ are ranked alphabets of \emph{input} and
    \emph{output symbols},
  \item $I \subseteq Q$ is a set of \emph{initial states},
  \item $R \subseteq T_\Sigma(Q) \times Q \times T_\Delta(Q)$
    is a finite set of \emph{rules} such that $\ell$~and~$r$ are linear
    in~$Q$ and $\var(r) \subseteq \var(\ell)$ for every $(\ell, q, r)
    \in R$, and
  \item $c \colon Q \to \Reg(\Sigma)$ assigns regular look-ahead to
    each state.
  \end{compactitem}
\end{definition}

It is worth noting that we assign look-ahead to each state, but in a
derivation we will only use the look-ahead if the state is deleted
(see Definition~\ref{df:Sem}).  Moreover, in contrast to other
definitions~\cite{malgrahopkni07,fulmalvog11}, we do not allow the
same state to occur several times (neither in the left- nor in the
right-hand side).  However, with the help of a simple renaming, each
traditional linear extended top-down tree transducer can be written in
our slightly more restrictive format.  Next, we recall some important
syntactic properties of our model.  To this end, let $M = (Q, \Sigma,
\Delta, I, R, c)$ be an l-xt${}^{\text R}$ in the following.  It is
\begin{compactitem}
\item a \emph{linear extended tree transducer} (without
  look-ahead)~[l-xt], if $c(q) = T_\Sigma$ for every $q \in Q$,
\item a \emph{linear top-down tree transducer with regular look
    ahead}~[l-t${}^{\text R}$] if $\ell \in \Sigma(Q)$ for every
  $(\ell, q, r) \in R$, 
\item a \emph{linear top-down tree transducer} (without look
  ahead)~[l-t] if it is both an l-xt and an l-t${}^{\text R}$,
\item \emph{$\varepsilon$-free} if $\ell \notin Q$ for every $(\ell,
  q, r) \in R$,
\item \emph{strict} if $r \notin Q$ for every $(\ell, q, r) \in R$,
\item a \emph{delabeling} if $\ell \in \Sigma(Q)$ and $r \in Q \cup
  \Delta(Q)$ for every $(\ell, q, r) \in R$, 
\item \emph{nondeleting} if $\var(r) = \var(\ell)$ for every $(\ell,
  q, r) \in R$, and
\item a \emph{finite-state relabeling}~[qr] if it is a nondeleting,
  strict delabeling l-t such that $\pos_p(\ell) = \pos_p(r)$ for 
  every $(\ell, q, r) \in R$ and $p \in \var(r)$.
\end{compactitem}
We note that each l-t${}^{\text R}$ and each l-t is automatically
$\varepsilon$-free.  To simplify the notation in examples and
illustrations, we sometimes write rules as~$\ell \stackrel
q\longrightarrow r$ instead of~$(\ell, q, r)$.  Similarly, we write
$\ell \stackrel{q, p} \longrightarrow r$ as a shorthand for the two
rules $\ell \stackrel q\longrightarrow r$ and $\ell \stackrel
p\longrightarrow r$.  Moreover, for every $p \in Q$ and $(\ell, q, r)
\in R$ we identify $\pos_p(\ell)$~and~$\pos_p(r)$ with their unique
element if the sets are non-empty.  Since $\ell$~and~$r$ are linear
in~$Q$, there is indeed at most one element in
$\pos_p(\ell)$~and~$\pos_p(r)$.  Finally, for every $q \in Q$, we let
\[ R_q = \{ \rho \in R \mid \rho = (\ell, q, r) \} \]
be the subset of~$R$ that contains all rules for state~$q$.

\begin{example}
  \label{ex:lXT}
  Let us consider the l-xt${}^{\text R}$ $M_1 = (Q, \Sigma, \Delta,
  \{\mathord{\star}\}, R, c)$ given by
  \begin{compactitem}
  \item $Q = \{\mathord{\star}, p, q, q', \mathord{\id},
    \mathord{\id}'\}$,
  \item $\Sigma = \{ \sigma^{(2)}, \gamma_1^{(1)}, \gamma_2^{(1)} \}
    \cup \Delta$,
  \item $\Delta = \{ \sigma_1^{(2)}, \sigma_2^{(2)},
    \gamma^{(1)}, \alpha^{(0)}\}$,
  \item $R$ contains exactly the following rules
    \begin{align*}
      \sigma_1(p, q) &\stackrel {\star,p}\longrightarrow \sigma_1(p, q) &
      \sigma_2(\mathord{\id}, \mathord{\id'})
      &\stackrel{p,q}\longrightarrow \sigma_2(\mathord{\id},
      \mathord{\id}') & \gamma_1(p) &\stackrel p\longrightarrow p \\
      \sigma(q, \mathord{\id}) &\stackrel q\longrightarrow q &
      \sigma(q', q) &\stackrel q\longrightarrow q &
      \gamma_2(q) &\stackrel q\longrightarrow q \\
      \gamma(\mathord{\id}) &\stackrel{\mathord{\id},
        \mathord{\id}'}\longrightarrow \gamma(\mathord{\id}) & \alpha
      &\stackrel {\mathord{\id}, \mathord{\id}'}\longrightarrow \alpha
    \end{align*}
  \item $c(q') = \{ t \in T_\Sigma \mid \pos_{\sigma_2}(t) \neq \emptyset \}$ and
    $c(f) = T_\Sigma$ for all other states~$f$.
  \end{compactitem}
  It can easily be checked that $c(q')$ is a regular tree language.
  The l-xt${}^{\text R}$~$M_1$ is an $\varepsilon$-free, nondeleting,
  delabeling top-down tree transducer with regular look-ahead.
\end{example}

Next, we recall the semantics of the l-xt${}^{\text R}$~$M = (Q,
\Sigma, \Delta, I, R, c)$, which is (mostly) given by synchronous
substitution.  While the links in a rule are implicit and established
due to occurrences of equal states, we need an explicit linking
structure for our sentential forms.  These links will form
dependencies that are used in a proof later on.  To encode the links,
we store a relation between positions of the input and output tree.
Let ${\cal L} = \{ D \mid D \subseteq \nat^* \times \nat^*\}$ be the
set of all link structures.  First, we define general sentential
forms.  Roughly speaking, a sentential form consists of an input tree
and an output tree, in which positions are linked.

\begin{definition}[{\protect{see~\cite[Section~3]{fulmalvog10}}}]
  \label{df:SentForm}
  \upshape An element~$\langle \xi, D, \zeta \rangle \in T_\Sigma(Q)
  \times \mathcal L \times T_\Delta(Q)$ is a \emph{sentential form}
  (for $M$) if $v \in \pos(\xi)$ and $w \in \pos(\zeta)$ for every
  $(v, w) \in D$. For a set $\cal S$ of sentential forms we define
  \[ \link({\cal S})=\{D\mid \langle \xi, D, \zeta \rangle \in {\cal S}
  \text{ for some trees } \xi \text{ and } \zeta\} \enspace. \]
\end{definition}

Now we formalize the implicit links in a rule~$\rho$ by defining the
explicit link structure determined by~$\rho$.  This explicit link
structure is added to the link structure of a sentential form whenever
the rule~$\rho$ is applied in the derivation process.

\begin{definition} \upshape
  \label{df:Links3}
  Let $\rho \in R$ be the rule~$(\ell, q, r)$, and let $v, w \in
  \nat^*$.  The explicit link structure~of $\rho$ for the positions
  $v$~and~$w$ is defined as
  \[ \link_{v, w}(\rho) = \{ (v.\pos_p(\ell), w.\pos_p(r)) \mid p \in
  \var(r) \} \enspace. \]
\end{definition}

\begin{example}
  \label{ex:lXTlinks}
  Let us compute two explicit link structures.
  \begin{align*}
    \link_{1, 21}(\sigma_1(p, q) \stackrel \star\longrightarrow
    \sigma_1(p, q)) &= \{(11, 211), (12, 212)\} \\
    \link_{1, 21}(\sigma(q', q) \stackrel q\longrightarrow q) &= \{(12,
    21)\}    
  \end{align*}
  In illustrations we use grayed splines to indicate the links.  The
  rules above and their implicit links [i.e., those
  of~$\link_{\varepsilon, \varepsilon}(\rho)$] are displayed in
  Figure~\ref{fig:links}. 
\end{example}

\begin{figure}[t]
  \centering
  \includegraphics{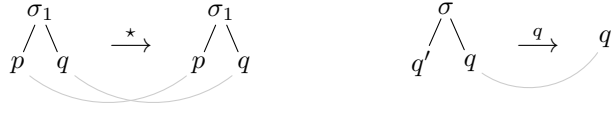}
  \caption{Illustration of two rules with their implicit links.}
  \label{fig:links}
\end{figure}

The derivation process is started with a simple sentential
form~$\langle q, \{(\varepsilon, \varepsilon)\}, q \rangle$ consisting
of the input and output tree~$q$ for some initial state $q \in I$ and
the trivial link relating both roots.  The current sentential form
can evolve in two ways.  Either we (nondeterministically) apply a
rule~$(\ell, q, r)$ to a pair $(v,w)$ of linked occurrences of the
state~$q$ or we apply the look-ahead.  In the former case, such a rule
application replaces the linked occurrences of~$q$ by the left and
right-hand side of the rule.  The explicit link structure of the rule
for $v$ and $w$ is added to the current (explicit) link structure to
obtain a new sentential form.  Since we are interested in the
dependencies created during derivation, we preserve all links and
never remove a link from the linking structure.  This replacement
process can be repeated until no linked occurrences of states remain.
Thus, we obtain an output tree without states, but the input tree of
the sentential form can still contain states, which do not take part
in an active link (i.e., a link relating two states in the sentential
form).  Each occurrence of such a state~$q$ can simply be replaced by
any tree of the regular look-ahead~tree language~$c(q)$, where
different occurrences of the same state can be replaced by different
trees of~$c(q)$.

\begin{definition}[{\protect{see~\cite[Section~3]{fulmalvog10}}}]
  \label{df:Sem}
  \upshape
  Given two sentential forms $\langle \xi, D, \zeta \rangle$ and
  $\langle \xi', D', \zeta' \rangle$, we write 
  \[ \langle \xi, D, \zeta \rangle \Rightarrow_M \langle \xi', D',
  \zeta' \rangle \] if 
  \begin{compactitem}
  \item there exists a rule $\rho = (\ell, q, r) \in R$ and input and
    output positions $(v, w) \in D \cap \bigl(\pos_q(\xi) \times
    \pos_q(\zeta) \bigr)$ such that
    \[ \xi' = \xi[\ell]_v \qquad \zeta' = \zeta[r]_w \qquad D' = D
    \cup \link_{v, w}(\rho) \enspace, \]
  \item or there exists a position $v \in \pos_Q(\xi)$ with $w
    \notin \pos_Q(\zeta)$ for all $(v, w) \in D$ and there exists $t
    \in c(\xi(v))$ such that
    \[ \xi' = \xi[t]_v \qquad \zeta' = \zeta \qquad D' = D \enspace. \]
  \end{compactitem}
\end{definition}

\begin{figure}
  \centering
  \includegraphics[scale=1]{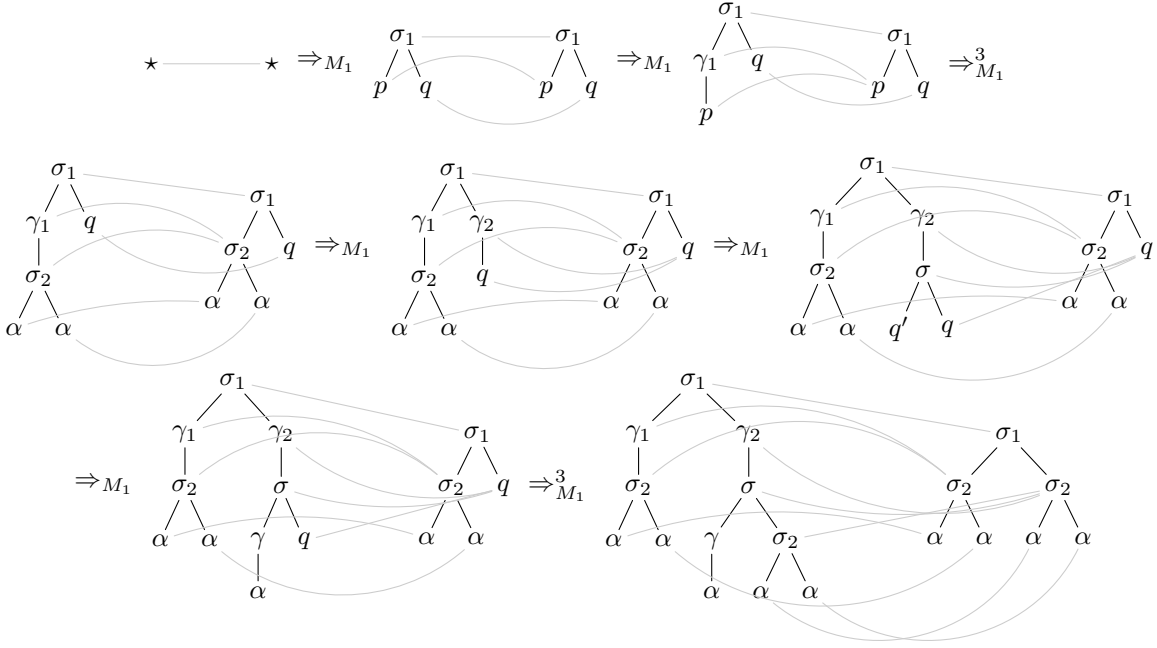}
  \caption{Derivation using the l-XT${}^{\text R}$~$M_1$ of
    Example~\protect{\ref{ex:lXT}}.}
  \label{fig:xDeriv}
\end{figure}

A few derivation steps using the l-xt${}^{\text R}$~$M_1$ of
Example~\ref{ex:lXT} are illustrated in Figure~\ref{fig:xDeriv}.  Next,
we define the tree transformation computed by an l-xt${}^{\text R}$.

\begin{definition} \upshape
  \label{df:Sem2}
  The l-xt${}^{\text R}$~$M$ computes the
  dependencies~$\dep(M)$, which are given by
  \[ \dep(M) = \{ (t, D, u) \in T_\Sigma \times \mathcal L \times T_\Delta
  \mid \exists q \in I \colon \langle q, \{(\varepsilon,
  \varepsilon)\}, q \rangle \Rightarrow_M^* (t, D, u) \} \enspace, \]
  where~$\Rightarrow_M^*$ is the reflexive and transitive closure
  of~$\Rightarrow_M$.  Moreover, it computes the tree
  transformation~$M$, which is given by $M = \{ (t, u) \mid (t, D, u)
  \in \dep(M)\}$.
\end{definition}

\begin{example}
  \label{ex:Sem2}
  Let $M_1$ be the l-xt${}^{\text R}$ of Example~\ref{ex:lXT}.
  Then
  \[ \bigl(\sigma_1(\gamma_1(\sigma_2(\alpha, \alpha)),
  \gamma_2(\sigma(\gamma(\alpha), \sigma_2(\alpha, \alpha)))), D,
  \sigma_1(\sigma_2(\alpha, \alpha), \sigma_2(\alpha, \alpha))\bigr) \in
  \dep(M_1) \] where 
  \begin{align*}
    D &= \{(\varepsilon, \varepsilon), (1, 1), (11, 1),
    (111, 11), (112, 12), 
    (2, 2), (21, 2), (212, 2), (2121, 21),
    (2122, 22)\} \enspace, 
  \end{align*}
  which corresponds to the final element in the derivation displayed
  in Figure~\ref{fig:xDeriv}.  Roughly speaking, the tree transformation
  computed by~$M_1$ accepts only input trees of a certain shape.  In
  such a tree it removes all $\gamma_1$-symbols on the left spine
  between the symbols $\sigma_1$~and~$\sigma_2$.  In addition, it also
  removes all symbols to the right of~$\sigma_1$ except for those
  belonging to the left-most subtree rooted by~$\sigma_2$.
\end{example}

Since every translation $(t, u) \in M$ is ultimately created by (at
least) one successful derivation, we can inspect the links in the
derivation process to exhibit the dependencies.  Roughly speaking,
the links establish which parts of the output tree were generated
due to a particular part of the input tree.  This correspondence is
called \emph{contribution} in~\cite{engman03b}.

\begin{definition} \upshape
  \label{df:classes}
  To allow concise statements, we introduce the following short-hands:
  \[ \not\!\varepsilon = \text{$\varepsilon$-free} \qquad
  \text s = \text{strict} \qquad
  \text d = \text{delabeling} \qquad
  \text n = \text{nondeleting} \enspace.
  \]
  We use these abbreviations in conjunction with l-xt${}^{\text R}$ to
  restrict to devices with the indicated properties.  For example,
  dnl-xt stands for ``delabeling and nondeleting linear extended
  top-down tree transducer'' (without look-ahead).  We use the same
  abbreviations with the stem (i.e., the material behind the hyphen)
  in capital letters for the corresponding classes of computed tree
  transformations.  For instance, dnl-XT stands for the class of all
  tree transformations computable by dnl-xt, and QR~denotes the class
  of all tree transformations computable by~qr.
\end{definition}

Our device develops the input and output tree from top to bottom,
which yields the name top-down tree transducer.  Similarly, there
exists a device called linear extended bottom-up tree
transducer~\cite{fulmalvog11}, which creates the trees from the leaves
to the root (i.e., bottom-up)~\cite{eng75}.  We do not recall the
device formally here, but we use the abbreviations l-xb~and~l-XB
together with the usual properties for it.  We refer the reader
to~\cite{fulmalvog11} for full details.

Next, we recall two important known statements for the class
l-XT${}^{\text R}$.  The first statement relates the classes
l-XT${}^{\text R}$~and~l-T${}^{\text R}$ to l-XB~and~l-B,
respectively.  In addition, we demonstrate that look-ahead is
superfluous in the nondeleting case.  We use the brackets `['~and~`]'
for optional use of the restrictions $\not\!\varepsilon$, `d', `s',
and `n' that have to be consistently applied.

\begin{theorem}[\protect{\cite[Proposition~3.3 and
    Corollary~4.1]{fulmalvog11}}] \upshape
  \label{thm:TB}
  \begin{align*}
    \text{[$\not\!\varepsilon$][s][d][n]l-XT}^{\text R} &=
    \text{[$\not\!\varepsilon$][s][d][n]l-XB} &
    \text{[s][d][n]l-T}^{\text R} &= \text{[s][d][n]l-B} \\
    \text{[$\not\!\varepsilon$][s][d]nl-XT}^{\text R} &=
    \text{[$\not\!\varepsilon$][s][d]nl-XT} &
    \text{[s][d]nl-T}^{\text R} &= \text{[s][d]nl-T}
  \end{align*}
\end{theorem}

To illustrate the consistent application of optional restrictions, the
following statements are instances of the first result of the above
theorem:
\[ \text{l-XT}^{\text R} = \text{l-XB} \qquad \text{and} \qquad
    \text{$\not\!\varepsilon$sl-XT}^{\text R} =
    \text{$\not\!\varepsilon$sl-XB} \enspace. \]

Secondly, we relate the class l-XT${}^{\text R}$ to l-T${}^{\text R}$, which
tells us how to emulate linear extended top-down tree transducers with
regular look-ahead by linear top-down tree transducers with regular
look-ahead~\cite{eng77}.

\begin{theorem} \upshape
  \label{thm:BBIM}
  \begin{align*}
    \text{$\not\!\varepsilon$[s][d][n]l-XT}^{\text R} &=
    \text{lcs-H}^{-1} \mathbin; \text{[s][d][n]l-T}^{\text R} &
    \text{[s][d][n]l-XT}^{\text R} &=
    \text{lc-H}^{-1} \mathbin; \text{[s][d][n]l-T}^{\text R}
  \end{align*}
\end{theorem}

\begin{proof}
  The statements are easily obtained from Theorem~\ref{thm:TB} and
  \cite[Lemma~4.1]{fulmalvog11}.
\end{proof}

\section{Four classes that are closed at a finite power}
\label{sec:upper}
In this section, we show that the classes $\text{$\mathord{\not{\!\varepsilon}}$l-XT}^{\text
  R}$, $\text{$\mathord{\not{\!\varepsilon}}$l-XT}$,
$\text{$\mathord{\not{\!\varepsilon}}$sl-XT}^{\text R}$, and
$\text{$\mathord{\not{\!\varepsilon}}$sl-XT}$ are closed under
composition at a finite power.   We first recall a
central result~\cite{arndau82} that proves that the (very restricted)
class $\text{$\mathord{\not{\!\varepsilon}}$snl-XT}$ is closed under
composition at power~2.  Note that \cite{arndau82}~expresses this
result in terms of bimorphisms, and `strict' is abbreviated
by~`e' there.  In fact,
$\text{$\mathord{\not{\!\varepsilon}}$snl-XT}=\text{B}(\text{lcs,
  lcs})$ by~\cite{arndau76} and \cite[Theorem~4]{mal07e}.

\begin{theorem}[{\protect{\cite[Theorem~6.2]{arndau82}}}] \upshape
  \label{thm:norm}
  For every $n \geq 2$,
  \[ \text{$\mathord{\not{\!\varepsilon}}$snl-XT} \subsetneq
  \text{$\mathord{\not{\!\varepsilon}}$snl-XT}^2 = 
  \text{$\mathord{\not{\!\varepsilon}}$snl-XT}^n \enspace. \]
\end{theorem}

Now we establish our first composition result, which is analogous to
the classical composition result for linear top-down tree transducers
with regular look-ahead~\cite{eng77}.  The only difference is that our
first transducer has extended left-hand sides (i.e., it is an
l-xt${}^{\text R}$ instead of just an l-t${}^{\text R}$).  Since this
fact does not affect the original composition
construction~\cite{eng75}, we can use the original result to obtain
our first result.

\begin{lemma}[composition on the right] \upshape
  \label{lm:top}
  \[ \text{[$\mathord{\not{\!\varepsilon}}$][s][d][n]l-XT}^{\text R}
  \mathbin; \text{[s][d][n]l-T}^{\text R} =
  \text{[$\mathord{\not{\!\varepsilon}}]$[s][d][n]l-XT}^{\text R} \]
\end{lemma}

\begin{proof}
  By Theorem~\ref{thm:BBIM},
  $\text{[$\mathord{\not{\!\varepsilon}}$][s][d][n]l-XT}^{\text R} =
  \text{lc[s]-H}^{-1} \mathbin; \text{[s][d][n]l-T}^{\text R}$, where
  the option~[$\mathord{\not{\!\varepsilon}}$] in the left-hand side
  corresponds to the option~[s] in~$\text{lc[s]-H}^{-1}$.  Thus, we
  obtain the chain of equalities
  \begin{align*}
    \text{[$\mathord{\not{\!\varepsilon}}$][s][d][n]l-XT}^{\text R}
    \mathbin; \text{[s][d][n]l-T}^{\text R} 
    &= \text{lc[s]-H}^{-1} \mathbin; \text{[s][d][n]l-T}^{\text R} \mathbin;
    \text{[s][d][n]l-T}^{\text R} \\
    &= \text{lc[s]-H}^{-1} \mathbin; \text{[s][d][n]l-T}^{\text R} \\
    &= \text{[$\mathord{\not{\!\varepsilon}}$][s][d][n]l-XT}^{\text R}
    \enspace,
  \end{align*}
  where in the second step we applied the following well-known
  composition result for linear top-down tree transducers with regular
  look-ahead~\cite{eng77}:\footnote{The abbreviation~`d' has a
    completely different meaning in~\protect{\cite{eng77}}.}
  \[ \text{[s][d][n]l-T}^{\text R} \mathbin;
  \text{[s][d][n]l-T}^{\text R} =
  \text{[s][d][n]l-T}^{\text R} \enspace. \tag*{\qed} \]
\end{proof}

In the next result we present a decomposition that corresponds to
property~P of~\cite{dau77}.  It demonstrates how to simulate an
l-xt${}^{\text R}$ by a delabeling l-d${}^{\text R}$ and an
$\text{$\mathord{\not{\!\varepsilon}}$snl-xt}$, for which we have
the composition closure result in Theorem~\ref{thm:norm}. 

\begin{lemma}[decomposition] \upshape
  \label{lm:p}
  \[ \text{$\mathord{\not{\!\varepsilon}}$[s]l-XT}^{\text R}
  \subseteq \text{[s]dl-T}^{\text R} \mathbin;
  \text{$\mathord{\not{\!\varepsilon}}$snl-XT} \]
\end{lemma}

\begin{proof}
  Let $M = (Q, \Sigma, \Delta, I, R, c)$ be an arbitrary
  $\varepsilon$-free l-xt${}^{\text R}$.  Moreover, let $m \in \nat$
  be such that $m \geq \abs{\var(r)}$ for every $(\ell, q, r) \in R$;
  i.e., the integer~$m$ is larger than the maximal number of states in
  the right-hand side of the rules.  For every rule $\rho = (\ell, q,
  r) \in R$ and non-state position $w \in \pos_\Sigma(\ell)$ in its
  left-hand side, let
  \[ \used_\rho(w) = \{ i \in \nat \mid wi \in \pos(\ell), \var(\ell|_{wi})
  \cap \var(r) \neq \emptyset\} \] be the indices of the direct
  subtrees below~$w$ in~$\ell$, which still contain states that occur
  in~$r$.  We construct a delabeling l-xt${}^{\text R}$
  \[ M_1 = (Q_1, \Sigma, R \cup \{@_i \mid 0 \leq i \leq m\}, I_1,
  R_1, c_1) \] such that
  \begin{compactitem}
  \item $Q_1 = \{ \langle \rho, w\rangle \mid \rho = (\ell, q, r) \in
    R, w \in \pos(\ell) \}$,  
  \item $\rk(\rho) = \abs{\used_\rho(\varepsilon)}$ for every rule
    $\rho \in R$ and $\rk(@_i) = i$ for every $0 \leq i \leq m$, 
  \item $I_1 = \{ \langle \rho, \varepsilon \rangle \mid q \in I, \rho
    \in R_q\}$, and
  \item for every rule $\rho = (\ell, q, r) \in R$ and non-state
    position $w \in \pos_\Sigma(\ell)$ we construct the following rule
    of~$R_1$:
    \[ \ell(w)(\langle \rho, w1\rangle, \dotsc, \langle \rho,
    wk\rangle) \stackrel{\langle \rho, w\rangle}\longrightarrow
    \begin{cases}
      \langle \rho, wi_1 \rangle & \text{if } r \in Q \\
      \rho(\langle \rho, wi_1\rangle, \dotsc, \langle \rho,
      wi_n\rangle) & \text{if } r \notin Q, w = \varepsilon \\
      @_n(\langle \rho, wi_1\rangle, \dotsc, \langle \rho,
      wi_n\rangle) & \text{otherwise,}
    \end{cases}
    \]
    where $k = \rk(\ell(w))$ and $\{\seq i1n\} = \used_\rho(w)$ with
    $i_1 < \dotsb < i_n$, 
  \item for every rule $\rho = (\ell, q, r) \in R$, state position~$w
    \in \pos_Q(\ell)$ and rule $\rho' \in R_{\ell(w)}$ we construct
    the following rule of~$R_1$:
    \[ \langle \rho', \varepsilon\rangle \stackrel{\langle \rho,
      w\rangle} \longrightarrow \langle \rho', \varepsilon\rangle
    \enspace, \]
  \item $c_1(\langle \rho, w\rangle) = (\ell|_w) c$ for every $\rho =
    (\ell, q, r) \in R$ and $w \in \pos(\ell)$.
  \end{compactitem}
  To obtain the desired l-t${}^{\text R}$ we simply eliminate the
  $\varepsilon$-rules using standard methods.  The elimination is
  successful because the $\varepsilon$-rules are
  non-cyclic.\footnote{Note that due to the $\varepsilon$-freeness
    of~$M$, we have $w \neq \varepsilon$ in the $\varepsilon$-rules of
    the second item.  Since these rules are the only constructed
    $\varepsilon$-rules, we cannot chain two $\varepsilon$-rules.}
  Intuitively speaking, the transducer $M_1$ processes the input and
  deletes subtrees that are not necessary for further processing.
  Moreover, it marks the positions in the input where a rule
  application would be possible.  Finally, it also already executes
  all erasing rules of~$M$.

  Let $m' \in \nat$ be such that $m' \geq \abs{\ell}$ for all $(\ell,
  q, r) \in R$.  We construct the l-xt
  \[ M_2 = (Q, R \cup \{@_i \mid 0 \leq i \leq m\}, \Delta, I, R_2) \]
  such that $R_2$~contains all valid rules $(\rho(\seq t1k), q, r)$
  with $\rho = (\ell, q, r) \in R$ of a strict ln-xt with $\pos_R(t_i)
  = \emptyset$ and $\abs{t_i} \leq m'$ for every $1 \leq i \leq k$,
  where $k = \rk(\rho)$. \qed
\end{proof}

\begin{example}
  \label{ex:decomp}
  A full example for the construction of Lemma~\ref{lm:p} would be
  quite lengthy, so let us only illustrate the construction of~$M_1$
  on the example rule~$\rho$:
  \[ \sigma(p, \sigma(\alpha, q)) \stackrel q\longrightarrow
  \sigma(\alpha, \sigma(q, \alpha)) \enspace, \]
  for which the only non-trivial look-ahead is~$c(p) = L$.  For this
  rule we construct the following rules in~$M_1$:
  \begin{align*}
    \sigma(\langle \rho, 1\rangle, \langle \rho, 2\rangle)
    &\stackrel{\langle \rho, \varepsilon\rangle}\longrightarrow
    \rho(\langle \rho, 2\rangle) &  & &
    \langle \rho', \varepsilon\rangle &\stackrel{\langle \rho,
      1\rangle} \longrightarrow \langle \rho', \varepsilon\rangle \\*
    \sigma(\langle \rho, 21\rangle, \langle \rho, 22\rangle)
    &\stackrel{\langle \rho, 2\rangle} \longrightarrow @_1(\langle \rho,
    22\rangle) & \alpha &\stackrel{\langle \rho, 21\rangle}
    \longrightarrow @_0 & \langle \rho'', \varepsilon\rangle
    &\stackrel{\langle \rho, 22\rangle} \longrightarrow \langle
    \rho'', \varepsilon\rangle
  \end{align*}
  for all rules $\rho' \in R_p$ and $\rho'' \in R_q$.  Moreover, the
  look-ahead~$c_1$ of~$M_1$ is such that
  \[ c_1(\langle \rho, 1\rangle) = L \qquad \qquad c_1(\langle \rho,
  21\rangle) = \{\alpha\} \enspace. \]
\end{example}

Now we are able to prove that the class
$\text{$\mathord{\not{\!\varepsilon}}$l-XT}^{\text R}$ is closed under
composition at the third power.

\begin{figure}[t]
  \centering
  \includegraphics[scale=1]{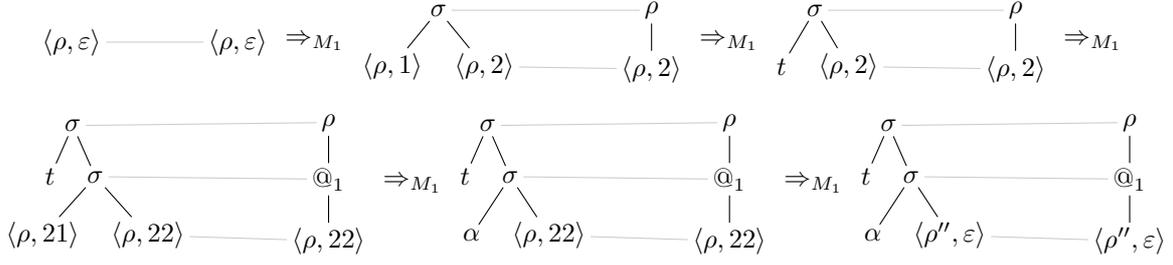}
  \caption{Derivation using~$M_1$ of Example~\ref{ex:decomp} for some
    $t \in L$ and $\rho'' \in R_q$.}
  \label{fig:xDeriv2}
\end{figure}

\begin{theorem} \upshape
  \label{thm:lR}
  For every $n \geq 1$,
  \[ (\text{$\mathord{\not{\!\varepsilon}}$[s]l-XT}^{\text R})^n
  \subseteq \text{[s]dl-T}^{\text R} \mathbin;
  \text{$\mathord{\not{\!\varepsilon}}$snl-XT}^2 \subseteq
  (\text{$\mathord{\not{\!\varepsilon}}$[s]l-XT}^{\text R})^3 \enspace. \]
\end{theorem}

\begin{proof}
  The second inclusion is trivial, so we prove the first inclusion by
  induction over~$n$.  The statement is a trivial consequence of
  Lemma~\ref{lm:p} for~$n = 1$.  In the induction step, we obtain
  \begin{align*}
    (\text{$\mathord{\not{\!\varepsilon}}$[s]l-XT}^{\text R})^{n+1} &=
    \text{$\mathord{\not{\!\varepsilon}}$[s]l-XT}^{\text R} \mathbin;
    (\text{$\mathord{\not{\!\varepsilon}}$[s]l-XT}^{\text R})^n \\
    &\subseteq \text{$\mathord{\not{\!\varepsilon}}$[s]l-XT}^{\text R}
    \mathbin; \text{[s]dl-T}^{\text R} \mathbin;
    \text{$\mathord{\not{\!\varepsilon}}$snl-XT}^2 \\
    &= \text{$\mathord{\not{\!\varepsilon}}$[s]l-XT}^{\text R} \mathbin;
    \text{$\mathord{\not{\!\varepsilon}}$snl-XT}^2 \\
    &\subseteq \text{[s]dl-T}^{\text R} \mathbin;
    \text{$\mathord{\not{\!\varepsilon}}$snl-XT} \mathbin;
    \text{$\mathord{\not{\!\varepsilon}}$snl-XT}^2 \\
    &= \text{[s]dl-T}^{\text R} \mathbin;
    \text{$\mathord{\not{\!\varepsilon}}$snl-XT}^2
  \end{align*}
  where we used the induction hypothesis in the second step,
  Lemma~\ref{lm:top} in the third step, Lemma~\ref{lm:p} in the fourth
  step, and finally, Theorem~\ref{thm:norm} in the last step. \qed
\end{proof}

It is known that we can simulate an l-t${}^{\text R}$ (with
look-ahead) by a composition of two l-t (without look-ahead). This fact
allows us to obtain the following corollary for the closure 
of the class~$\text{$\mathord{\not{\!\varepsilon}}$l-XT}$ under composition at power four.

\begin{corollary}[without look-ahead] \upshape
  \label{cor:l}
  For every $n \geq 1$
  \[ \text{$\mathord{\not{\!\varepsilon}}$[s]l-XT}^n \subseteq
  \text{QR} \mathbin; \text{[s]dl-T} \mathbin;
  \text{$\mathord{\not{\!\varepsilon}}$snl-XT}^2 \subseteq
  \text{$\mathord{\not{\!\varepsilon}}$[s]l-XT}^4 \enspace. \]
\end{corollary}

\begin{proof}
  The second inclusion is trivial, and for the first inclusion we use
  $\text{[s]dl-T}^{\text R} \subseteq \text{QR}
  \mathbin; \text{[s]dl-T}$ and Theorem~\ref{thm:lR}.  \qed
\end{proof}

Up to now, we have shown that the classes
$\text{$\mathord{\not{\!\varepsilon}}$l-XT}^{\text
  R}$~and~$\text{$\mathord{\not{\!\varepsilon}}$l-XT}$ are closed
under composition at the third and fourth power, respectively.  In the
rest of the section, we will show that the (strict) classes
$\text{$\mathord{\not{\!\varepsilon}}$sl-XT}^{\text R}$ and
$\text{$\mathord{\not{\!\varepsilon}}$sl-XT}$ are closed under
composition already at the second power.  We start with a lemma that
proves the converse of Lemma~\ref{lm:p} in the strict case.

\begin{lemma}[composition on the left] \upshape
  \label{lm:del}
  \[ \text{lds-H} \mathbin;
  \text{$\mathord{\not{\!\varepsilon}}$sl-XT}^{\text R} \subseteq
  \text{$\mathord{\not{\!\varepsilon}}$sl-XT}^{\text R} \]
\end{lemma}

\begin{proof}
  Let $d \colon T_\Sigma \to T_\Delta$ be a strict delabeling linear
  tree homomorphism.  Recall that~$d$ also defines a tree
  transformation $d \colon T_\Sigma(Q) \to T_\Delta(Q)$, which acts as
  an identity on states; i.e., $d(q) = q$ for every $q \in Q$.
  Moreover, let $M = (Q, \Delta, \Gamma, I, R, c)$ be a strict
  l-xt${}^{\text R}$ and $m \in \nat$ be such that $m \geq
  \abs{\ell'}$ for every $\ell' \in d^{-1}(\ell)$ and $(\ell, q, r)
  \in R$.  Note that the set~$d^{-1}(\ell)$ is finite because $d$~is
  strict.  We construct the l-xt${}^{\text R}$ $M' = (Q \cup \{1,
  \dotsc, m\}, \Sigma, \Gamma, I, R', c')$ such that for every rule
  $(\ell, q, r) \in R$ we have each valid rule $(\ell', q, r)$ in~$R'$
  where $\ell' \in d^{-1}(\ell)$ and $\abs{\pos_\Sigma(\ell')} =
  \abs{\pos_\Delta(\ell)}$.  Moreover, $c'(q) = d^{-1}(c(q))$ for all
  $q \in Q$ and $c'(i) = T_\Sigma$ for every $1 \leq i \leq m$.
  Clearly, $d^{-1}(c(q))$ is a regular tree language because $c(q)$~is
  a regular tree language and regular tree languages are closed under
  inverse tree homomorphisms~\cite[Lemma~1.2]{eng77}.  Finally, we
  observe that $M'$~is also strict since it has the same right-hand
  sides as~$M$. \qed
\end{proof}

Now we state and prove the mentioned result for the class
$\text{$\mathord{\not{\!\varepsilon}}$sl-XT}^{\text R}$.  The general
approach is the same as in Theorem~\ref{thm:lR}, but we now use
Lemma~\ref{lm:del} in the main step.

\begin{theorem} \upshape
  \label{thm:lsR}
  For every $n \geq 1$
  \[ (\text{$\mathord{\not{\!\varepsilon}}$sl-XT}^{\text R})^n
  \subseteq 
  \text{$\mathord{\not{\!\varepsilon}}$snl-XT} \mathbin;
  \text{$\mathord{\not{\!\varepsilon}}$sl-XT}^{\text R} \subseteq
  (\text{$\mathord{\not{\!\varepsilon}}$sl-XT}^{\text R})^2 \enspace. \]
\end{theorem}

\begin{proof}
  Again, the second inclusion is trivial.  For the first inclusion,
  the induction basis is also trivial, and we prove the induction step
  as follows
  \begin{align*}
    (\text{$\mathord{\not{\!\varepsilon}}$sl-XT}^{\text R})^{n+1} &=
    (\text{$\mathord{\not{\!\varepsilon}}$sl-XT}^{\text R})^n \mathbin;
    \text{$\mathord{\not{\!\varepsilon}}$sl-XT}^{\text R} \\
    &\subseteq \text{$\mathord{\not{\!\varepsilon}}$snl-XT}
    \mathbin; \text{$\mathord{\not{\!\varepsilon}}$sl-XT}^{\text R}
    \mathbin; \text{$\mathord{\not{\!\varepsilon}}$sl-XT}^{\text R} \\
    &= \text{$\mathord{\not{\!\varepsilon}}$snl-XT} \mathbin;
    \text{$\mathord{\not{\!\varepsilon}}$snl-XT} \mathbin; \text{lds-H}
    \mathbin; \text{$\mathord{\not{\!\varepsilon}}$sl-XT}^{\text R} \\
    &\subseteq \text{$\mathord{\not{\!\varepsilon}}$snl-XT}^2
    \mathbin; \text{$\mathord{\not{\!\varepsilon}}$sl-XT}^{\text R} \\
    &= \text{$\mathord{\not{\!\varepsilon}}$snl-XT}^2
    \mathbin; \text{$\mathord{\not{\!\varepsilon}}$snl-XT}
    \mathbin; \text{lds-H} \\
    &= \text{$\mathord{\not{\!\varepsilon}}$snl-XT}^2
    \mathbin; \text{lds-H} \\
    &= \text{$\mathord{\not{\!\varepsilon}}$snl-XT}
    \mathbin; \text{$\mathord{\not{\!\varepsilon}}$sl-XT}^{\text R}
  \end{align*}
  using the induction hypothesis in the second step,
  Lemma~\ref{lm:del} in the fourth step, and 
  Theorem~\ref{thm:norm} in the sixth step.  The final step composes
  the tree homomorphism with the second tree transducer using the result of
  Lemma~\ref{lm:top}. \qed
\end{proof}

Again, we derive the ``without look-ahead''-version of the general result.  In contrast to Corollary~\ref{cor:l},
the power of closedness does not increase by one for
strict~$\text{$\mathord{\not{\!\varepsilon}}$l-xt}$.

\begin{corollary}[without look-ahead] \upshape
  \label{cor:ls}
  For every $n \geq 1$
  \[ \text{$\mathord{\not{\!\varepsilon}}$sl-XT}^n \subseteq
  \text{$\mathord{\not{\!\varepsilon}}$snl-XT} \mathbin;
  \text{$\mathord{\not{\!\varepsilon}}$sl-XT} \subseteq
  \text{$\mathord{\not{\!\varepsilon}}$sl-XT}^2 \enspace. \]
\end{corollary}

\begin{proof}
  The second inclusion is trivial.  For the first inclusion, we
  observe that 
  \begin{align*}
    \text{$\mathord{\not{\!\varepsilon}}$sl-XT}^n &\subseteq
    \text{$\mathord{\not{\!\varepsilon}}$snl-XT} \mathbin;
    \text{$\mathord{\not{\!\varepsilon}}$sl-XT}^{\text{R}} \\
    &= \text{lcs-H}^{-1} \mathbin; \text{snl-T} \mathbin; \text{QR}
    \mathbin; \text{$\mathord{\not{\!\varepsilon}}$sl-XT} \\
    &= \text{lcs-H}^{-1} \mathbin; \text{snl-T} \mathbin;
    \text{$\mathord{\not{\!\varepsilon}}$sl-XT} \\
    &= \text{$\mathord{\not{\!\varepsilon}}$snl-XT}
    \mathbin; \text{$\mathord{\not{\!\varepsilon}}$sl-XT}
  \end{align*}
  using first Theorem~\ref{thm:lsR}, and then Theorem~\ref{thm:BBIM}
  and~\cite[Theorem~2.6]{eng77} in the second step.  Next, we use the
  composition result~\cite[Theorem~1]{bak79} for top-down tree
  transducers and Theorem~\ref{thm:BBIM} again. \qed
\end{proof}

Table~\ref{tab:max} summarizes our results concerning the powers at
which the considered classes are closed under composition.

\begin{table}[t]
  \begin{center}
    \begin{tabular}{|l|c|c|}
      \hline
      & & \\[-1ex]
      \; Class & \mbox{\quad} Closed at power \mbox{\quad} & Proved in \\[1ex]
      \hline
      & & \\[-1ex]
      \; $\text{$\mathord{\not{\!\varepsilon}}$l-XT}^{\text R}$ & ${} 
      3$ & Theorem~\ref{thm:lR} \\[1ex]
      \; $\text{$\mathord{\not{\!\varepsilon}}$l-XT}$ & ${}  4$ &
      Corollary~\ref{cor:l} \\[1ex]
      \; $\text{$\mathord{\not{\!\varepsilon}}$sl-XT}^{\text R}$ \quad &
      ${} 2$ & Theorem~\ref{thm:lsR} \\[1ex] 
      \; $\text{$\mathord{\not{\!\varepsilon}}$sl-XT}$ & ${} 2$ & \;
      Corollary~\ref{cor:ls} \; \\[1ex]
      \hline
    \end{tabular}
  \end{center}
  \caption{Summary of the results of Section~\protect{\ref{sec:upper}}.}
  \label{tab:max}
\end{table}

\section{Least power of closedness}
\label{sec:lower}
In this section, we will determine the least power at which each of the
classes $\text{$\mathord{\not{\!\varepsilon}}$l-XT}^{\text R}$,
$\text{$\mathord{\not{\!\varepsilon}}$l-XT}$,
$\text{$\mathord{\not{\!\varepsilon}}$sl-XT}^{\text R}$, and
$\text{$\mathord{\not{\!\varepsilon}}$sl-XT}$ is closed under
composition.

\begin{theorem} \upshape
  \label{thm:min1}
  For every $n \geq 2$
  \[ \text{$\mathord{\not{\!\varepsilon}}$sl-XT} \subsetneq
  \text{$\mathord{\not{\!\varepsilon}}$sl-XT}^2 =
  \text{$\mathord{\not{\!\varepsilon}}$sl-XT}^n \qquad 
  (\text{$\mathord{\not{\!\varepsilon}}$sl-XT}^{\text R}) \subsetneq
  (\text{$\mathord{\not{\!\varepsilon}}$sl-XT}^{\text R})^2 =
  (\text{$\mathord{\not{\!\varepsilon}}$sl-XT}^{\text R})^n .\]
\end{theorem}

\begin{proof}
  By \cite[Section~3.4]{arndau82}, the classes
  $\text{$\mathord{\not{\!\varepsilon}}$sl-XT}$ and
  $\text{$\mathord{\not{\!\varepsilon}}$sl-XT}^{\text R}$ are not
  closed under composition at power~$1$.  Moreover, by
  Theorem~\ref{thm:lsR} and Corollary~\ref{cor:ls} both classes are
  closed at power 2. \qed
\end{proof}

In the following, we will use the computed dependencies,
for which we observe some important properties next.  For the rest of
this section, let $M = (Q, \Sigma, \Delta, I, R, c)$ be the considered
l-xt${}^{\text R}$.  To simplify the development, we disregard the
actual input and output trees in the computed dependencies and will
consider only the set~$\link(M) = \link(\dep(M))$.  We say that
$M$~computes the linking structures of~$\link(M)$.

\begin{definition}[\protect{\cite[Definition~8]{mal11d}}] \upshape
  \label{df:Hier}
  A linking structure~$D \in {\cal L}$ is \emph{input
    hierarchical}\footnote{This notion is called \emph{strictly input
      hierarchical} in~\protect{\cite{mal11d}}.} if for all links
  $(v_1, w_1), (v_2, w_2) \in D$
  \begin{compactitem}
  \item if $v_1 \prec v_2$, then $w_1 \preceq w_2$, and
  \item if $v_1 = v_2$, then $w_1 \preceq w_2$~or~$w_2 \preceq w_1$.
  \end{compactitem}
\end{definition}

Input hierarchical linking structures have no crossing links, which
are links $(v_1, w_1), (v_2, w_2) \in D$ such that $v_1 \prec v_2$ and
$w_2 \prec w_1$.  We define \emph{output hierarchical} using the same
conditions as in Definition~\ref{df:Hier} for the output side; i.e.,
$D$~is output hierarchical if and only if $D^{-1}$~is input
hierarchical.  Moreover, $D$~is \emph{hierarchical} if it is both
input and output hierarchical.  Finally, a set~${\cal D} \subseteq
{\cal L}$ of linking structures is hierarchical if each element~$D \in
{\cal D}$ is hierarchical.

We also need a property that enforces the existence of certain links.
Roughly speaking, for a set of linking structures~${\cal D} \subseteq
{\cal L}$ there should be an integer that limits the distance between
links in each linking structure~$D \in {\cal D}$.

\begin{definition}[\protect{\cite[Definition~10]{mal11d}}] \upshape
  \label{df:Dist}
  A set~${\cal D} \subseteq {\cal L}$ of link structures has
  \emph{bounded distance in the input} if there exists an integer~$k
  \in \nat$ such that for every $D \in {\cal D}$ and all $(v, w),
  (vv'', w'') \in D$ there exists $(vv', w') \in D$ with $v' \preceq
  v''$ and $\abs{v'} \leq k$.
\end{definition}

In other words, in a set of link structures with bounded distance~$k$
in the input, we know for any of its link structures that between any
two source-nested links there exists a link such that the distance of
its source to the smaller link's source is at most~$k$.  This yields
that the distance to the source of the next nested link (if such a
link does exist) can be at most~$k$.  Note however, that the above
property does not require a link every $k$~symbols.  As before, a
set~${\cal D} \subseteq {\cal L}$ of linking structures has
\emph{bounded distance in the output} if ${\cal D}^{-1} = \{ D^{-1}
\mid D \in {\cal D}\}$ has bounded distance in the input.  Finally,
${\cal D}$~has \emph{bounded distance} if it has bounded distance in
both the input and the output.

\begin{lemma} \upshape
  \label{lm:Dep}
  For every l-xt${}^{\text R}$~$M$, the set~$\link(M)$ is
  hierarchical with bounded distance.
\end{lemma}

\begin{proof}
  This lemma follows trivially from Definition~\ref{df:Sem}.  
\end{proof}

Next we consider the problem whether a tree transformation can be
computed by an l-xt${}^{\text R}$.  For this we specify certain links
that are intuitively clear and necessary between nodes of input-output
tree pairs.  Hereby we obtain the specification sentential forms.
Then we consider whether this specification can be implemented by
an~l-xt${}^{\text R}$.  Often we cannot identify the nodes of a link
exactly.  In such cases, we use splines with inverted arrow heads,
which indicate that there is a link to some position of the subtree
pointed to.  For example, the splines in Figure~\ref{fig:compat}
indicate that a node of~$t_1$ (resp.~$t_2$) on the left is linked to a
node of~$t_1$ (resp.~$t_2$) on the right.

\begin{figure}[t]
  \centering
  \includegraphics{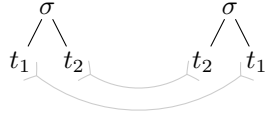}
  \caption{Links with inverted arrow.}
  \label{fig:compat}
\end{figure}

\begin{figure}[t]
  \centering
  \includegraphics[scale=1]{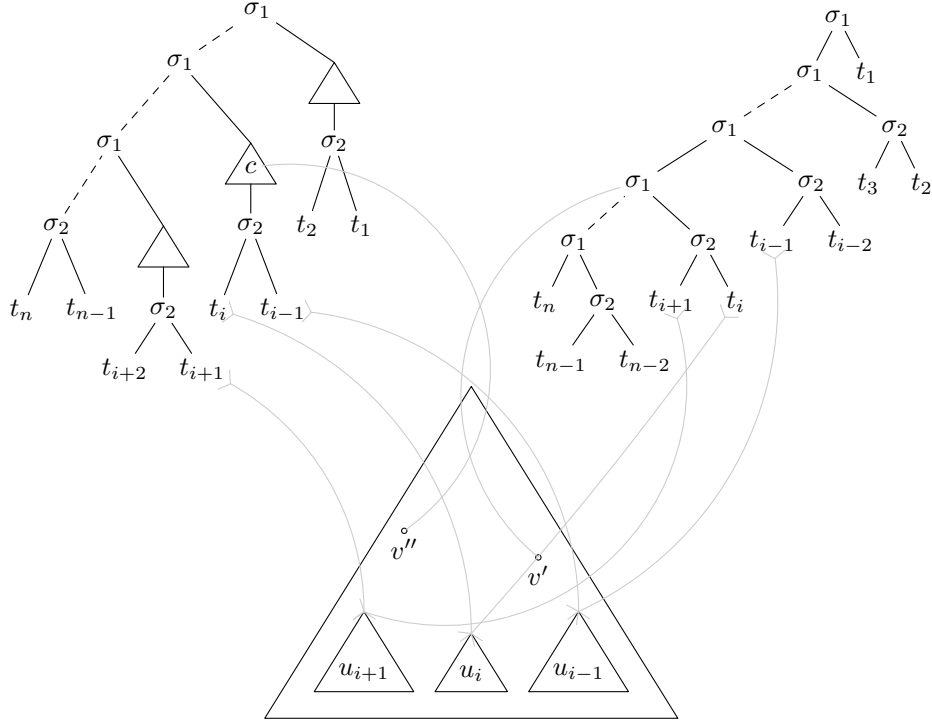}
  \caption{Illustration of the relevant part of the specification used
    in the proof of Theorem~\protect{\ref{thm:min2}}.}
  \label{fig:min2}
\end{figure}

\begin{theorem} \upshape
  \label{thm:min2}
  $(\text{$\mathord{\not{\!\varepsilon}}$l-XT}^{\text R})^2 \subsetneq
  (\text{$\mathord{\not{\!\varepsilon}}$l-XT}^{\text R})^3 =
  (\text{$\mathord{\not{\!\varepsilon}}$l-XT}^{\text R})^n$ for every
  $n \geq 3$.
\end{theorem}

\begin{proof} 
   The inclusion is trivial and the equality follows
  from~Theorem~\ref{thm:lR}, so we only have to prove strictness.
  Recall the l-xt${}^{\text R}$~$M_1$ of
  Example~\ref{ex:lXT}.  In addition, we use the two bimorphisms
  $B_2$~and~$B_3$ of~\cite[Section~II.2.2.3.1]{dau77}, which are in
  B(lcs,lcs).\footnote{In~\protect{\cite{dau77}} strictness is denoted
    by~`e'.}  As mentioned,
  $\text{$\mathord{\not{\!\varepsilon}}$snl-XT} = \text B(\text{lcs},
  \text{lcs})$, hence $B_2$~and~$B_3$ can also be defined by some
  $\text{$\mathord{\not{\!\varepsilon}}$snl-xt}$ $M_2$~and~$M_3$,
  respectively.  For convenience, we present $M_2$~and~$M_3$
  explicitly before we show that $\tau = M_1 \mathbin; M_2 \mathbin;
  M_3$ cannot be computed by a composition of two
  $\text{$\mathord{\not{\!\varepsilon}}$l-xt}^{\text R}$.

  Let $M_2 = (Q_2, \Delta, \Gamma, \{\mathord{\star}\}, R_2)$ be the
  $\text{$\mathord{\not{\!\varepsilon}}$snl-xt}$ with
  \begin{compactitem}
  \item $Q_2 = \{\mathord{\star}, \mathord{\id}, \mathord{\id}'\}$ and
    $\Gamma = \{ \sigma^{(2)}, \gamma^{(1)}, \alpha^{(0)}\}$, and
  \item $R_2$ is the set of the rules
    \begin{align*}
      \sigma_1(\mathord{\star}, \sigma_2(\mathord{\id},
      \mathord{\id}')) &\stackrel{\star}\longrightarrow
      \sigma(\sigma(\mathord{\star}, \mathord{\id}), \mathord{\id}') &
      \sigma_2(\mathord{\id}, \mathord{\id}')
      &\stackrel{\star}\longrightarrow \sigma(\mathord{\id},
      \mathord{\id}') \\*
      \gamma(\mathord{\id}) & \stackrel{\mathord{\id},
        \mathord{\id}'}\longrightarrow \gamma(\mathord{\id}) &
      \alpha & \stackrel{\mathord{\id}, \mathord{\id}'}\longrightarrow
      \alpha \enspace. 
    \end{align*}
  \end{compactitem}
  Moreover, let $M_3 = (Q_3, \Gamma, \Delta, \{\mathord{\star}\},
  R_3)$ be the $\text{$\mathord{\not{\!\varepsilon}}$snl-xt}$ with
  \begin{compactitem}
  \item $Q_3 = \{\mathord{\star}, p, \mathord{\id}, \mathord{\id}'\}$, 
  \item $R_3$~is the set of the rules
    \begin{align*}
      \sigma(p, \mathord{\id}) &\stackrel{\star}\longrightarrow
      \sigma_1(p, \mathord{\id}) &  
      \sigma(\sigma(p, \mathord{\id}), \mathord{\id}') &\stackrel
      p\longrightarrow \sigma_1(p, \sigma_2(\mathord{\id},
      \mathord{\id}')) &
      \gamma(\mathord{\id}) &\stackrel{\mathord{\id},
        \mathord{\id}'}\longrightarrow \gamma(\mathord{\id}) \\*
      & &
      \gamma(\mathord{\id}) &\stackrel p\longrightarrow
      \gamma(\mathord{\id}) &
      \alpha & \stackrel{\mathord{\id}, \mathord{\id}'}\longrightarrow
      \alpha \enspace.
    \end{align*}
  \end{compactitem}

  We present a proof by contradiction, hence we assume that $\tau =
  N_1 \mathbin;N_2$ for some
  $\text{$\mathord{\not{\!\varepsilon}}$l-xt}$ $N_1$~and~$N_2$.  With
  the help of Lemma~\ref{lm:Dep}, we conclude that
  $\link(N_1)$~and~$\link(N_2)$ are hierarchical with bounded distance
  $k_1$~and~$k_2$, respectively.  We consider $(t, u) \in \tau$ such
  that the left $\sigma_1$-spines of $t$~and~$u$ are longer than
  $k_1$~and~$k_2$, respectively.\footnote{Some additional requirements
    on $t$~and~$u$ are developed in the proof later on.} Moreover, by
  assumption, there exists an intermediate tree~$s$ and linking
  structures $D_1, D_2 \in {\cal L}$ such that $(t, D_1, s) \in
  \dep(N_1)$ and $(s, D_2, u) \in \dep(N_2)$.  We specify some links
  that, due to the last two inclusions, should be necessarily present
  in $D_1$~and~$D_2$ and show that they lead to a contradiction (see
  Figure~\ref{fig:min2}).

  First, we consider the links~$D_2$ generated by~$N_2$.  Since the
  left $\sigma_1$-spine in~$u$ is longer than~$k_2$ and there are
  links at the root (i.e., $(\varepsilon, \varepsilon) \in D_2$) and
  inside~$t_n$, there must be a linking point at position $w \in
  \pos_{\sigma_1}(u)$ along the left $\sigma_1$-spine with $w \neq
  \varepsilon$, which links to position~$v'$ in the intermediate
  tree~$s$ (i.e., $(v', w) \in D_2$).  Let $u|_w = \sigma_1(u',
  \sigma_2(t_{i+1}, t_i))$.  We assume that every~$t_j$ with $j \leq
  n$ is high enough (i.e., $\height(t_j) > k_2$).  This yields that
  there exist linking points in~$t_{i+1}$, $t_i$, and~$t_{i-1}$
  of~$u$.  Let $(v'_{i+1}, w_{i+1}), (v'_i, w_i), (v'_{i-1}, w_{i-1})
  \in D_2$ be those linking points.  Since $D_2$~is hierarchical and
  $w_{i+1}$~and~$w_i$ are below~$w$ in~$u$, we know that
  $v'_{i+1}$~and~$v'_i$ are below~$v'$ in~$s$ (i.e., $v' \preceq
  v'_{i+1}, v'_i$), whereas $v' \not\preceq v'_{i-1}$.  Next, we
  locate~$t_i$ in the input tree~$t$.  By the general shape of~$t$,
  the subtree~$t_i$ occurs in a subtree $\sigma_1(t', c[\sigma_2(t_i,
  t_{i-1})])$ for some tree $c \in T_\Sigma(X_1)$ with exactly one
  occurrence of~$x_1$.  Again we assume that $c$~is suitably
  large\footnote{More precisely, we assume that the only element
    of~$\pos_{x_1}(c)$ is longer than~$k_1$.}, which forces a linking
  point inside~$c$ in addition to those in~$t_{i+1}$, $t_i$,
  and~$t_{i-1}$ because $(t, s) \in N_1$.  Using the same arguments
  for~$N_1$ as before for~$N_2$, we now locate the links $(y, v'') \in
  D_1$ linking into~$c$, which dominates the links $(y_i, v''_i),
  (y_{i-1}, v''_{i-1}) \in D_1$ linking to $t_i$~and~$t_{i-1}$,
  respectively.  Thus, $v'' \preceq v''_i, v''_{i-1}$.  In addition,
  there is a link $(y_{i+1}, v''_{i+1}) \in D_1$ linking
  into~$t_{i+1}$, which is not dominated by~$v''$.  Consequently, $v''
  \not\preceq v''_{i+1}$.  Finally, let
  \[ v_{i-1} = \lca(v'_{i-1}, v''_{i-1}) \qquad \qquad v_i =
  \lca(v'_i, v''_i) \qquad \qquad v_{i+1} = \lca(v'_{i+1},
  v''_{i+1}) \enspace. \]
  It is easy to see that the relevant links can be chosen such that
  \begin{align*}
    v' &\not\preceq v_{i-1} & v' &\preceq v_i & v' &\preceq v_{i+1} \\
    v'' &\preceq v_{i-1} & v'' &\preceq v_i & v'' &\not\preceq v_{i+1}
    \enspace.
  \end{align*}
  Now, we have either $\lca(v_{i-1}, v_i) \preceq \lca(v_i,
  v_{i+1})$ or  $\lca(v_i, v_{i+1}) \preceq \lca(v_{i-1},
  v_i)$.  In the first case we get
  \[ v'' \preceq \lca(v_{i-1}, v_i) \preceq \lca(v_i, v_{i+1}) \preceq
  v_{i+1} \enspace, \]
  and in the second case we get
  \[ v' \preceq \lca(v_i, v_{i+1}) \preceq \lca(v_{i-1},
  v_i) \preceq v_{i-1} \enspace, \]
  which are contradictions. \qed
\end{proof}

Fortunately, we can use the proof of the previous theorem to conclude
that four is the least power at which the class
$\mathord{\not{\!\varepsilon}}$l-XT is closed under composition.

\begin{corollary} \upshape
  \label{cor:min3}
  $\text{$\mathord{\not{\!\varepsilon}}$l-XT}^3 \subsetneq
  \text{$\mathord{\not{\!\varepsilon}}$l-XT}^4 =
  \text{$\mathord{\not{\!\varepsilon}}$l-XT}^n$ for every $n \geq 4$.
\end{corollary}

\begin{proof} The inclusion is trivial and the equality follows from
  Corollary~\ref{cor:l}.  For the strictness,
  the proof of Theorem~\ref{thm:min2} essentially shows that in the
  first step we must delete the contexts indicated by triangles (such
  as~$c$) in Figure~\ref{fig:min2} because otherwise we can apply the
  method used in the proof to derive a contradiction (it relies on the
  existence of a linking point inside such a context~$c$).  Thus, in
  essence we must first implement a variant of the l-xt${}^{\text
    R}$~$M_1$ of Example~\ref{ex:lXT}.  It is a simple exercise to
  show that the deletion of the excess material cannot be done by a
  single l-xt as it cannot reliably determine the left-most occurrence
  of~$\sigma_2$ without the look-ahead.  Thus, if we only have l-xt to
  achieve the transformation, then we already need a composition of
  two~l-xt to perform the required deletion of the contexts, which
  proves the main statement. \qed
\end{proof}

In Table~\ref{tab:min} we summarize the results of this and the
previous section, which allow us to present the least power at which
the closure of the considered composition hierarchies is achieved.
For the sake of completeness, we also present the corresponding
results for the classes
$\text{$\mathord{\not{\!\varepsilon}}$snl-XT}$~and~$\text B(\text l,
\text l)$ that were obtained in~\cite{arndau82,dau77}.  Recall that
$\text B(\text l, \text l)$ is the class of all tree transformations
computable by bimorphisms, in which both tree homomorphisms are
linear.

\begin{table}[t]
  \begin{center}
    \begin{tabular}{|l|c|c|}
      \hline
      & & \\[-1ex]
      \; Class & \mbox{\quad} Least power of closedness \mbox{\quad} & Proved in
      \\[1ex] 
      \hline
      & & \\[-1ex]
      \; $\text B(\text l, \text l)$ & $4$ & 
      \protect{\cite[Section~II-2-2-3-3]{dau77}} \\[1ex] 
      \; $\text{$\mathord{\not{\!\varepsilon}}$snl-XT}=\text{$\mathord{\not{\!\varepsilon}}$snl-XT}^{\text R}$ \; & $2$ &
      \protect{\cite[Theorem~6.2]{arndau82}} \\[1ex]
      \hline
      & & \\[-1ex]
      \; $\text{$\mathord{\not{\!\varepsilon}}$sl-XT}^{\text R}$,  $\text{$\mathord{\not{\!   \varepsilon}}$sl-XT}$ & $2$ &
      Theorem \ref{thm:min1} \\[1ex]
      \; $\text{$\mathord{\not{\!\varepsilon}}$l-XT}^{\text R}$ & $3$ &
      Theorem \ref{thm:min2} \\[1ex]
      \; $\text{$\mathord{\not{\!\varepsilon}}$l-XT}$ & $4$ &
      Corollary \ref{cor:min3} \\[1ex] 
      \hline
    \end{tabular}
  \end{center}
  \caption{Summary of the results of Sections
    \protect{\ref{sec:upper}~and~\ref{sec:lower}}.} 
  \label{tab:min}
\end{table}

\section{Infinite composition hierarchies}
\label{sec:complete}
To complete the picture, we will need one further result showing the
infiniteness of the composition hierarchy for a large number of
classes.  In order to obtain a result that is as general as possible,
we use bimorphisms~\cite{arndau82} instead of l-xt${}^{\text R}$ in
this section.  We conclude several results for various tree transducer
classes from the result for bimorphisms.

The main auxiliary notion used in the proof of the infiniteness of the
composition hierarchy is a notion assigning levels to positions in a
tree.  Let $t \in T_\Sigma$ and $\ell \in \nat$ be an arbitrary
integer.  Since branching positions (i.e., those that are labeled by
symbols of rank at least~$2$) will play an essential role, we define
the set of all branching positions of~$t$ together with two different
successors and the branching positions along a given path.  For every
$w, ww'' \in \pos(t)$, let
\begin{align*}
  \branch_t &= \{ \langle w, i, j\rangle \mid w \in \pos(t),\, t(w)
  \notin \Sigma_0 \cup \Sigma_1,\, 1 \leq i, j \leq \rk(t(w)),\, i
  \neq j\} \\
  \branch_t(w, w'') &= \{ ww' \mid w' \preceq w'' \text{ and } (ww',
  i, j) \in \branch_t \text{ for some } i,j  \} \enspace. 
\end{align*}
We inductively define the sets~$\SSP^\ell_n(t) \subseteq \pos(t)
\times \nat \times \nat$ and $\SP^\ell_n(t) \subseteq \pos(t)$ for
every $t \in T_\Sigma$ and $n \in \nat$ as follows:
\begin{align*}
  \SSP^\ell_0(t) &= \branch_t \qquad \text{and} \qquad \SP^\ell_0(t) =
  \{ w \mid \langle w, i, j\rangle \in \SSP^\ell_0(t) \} \\*[1ex]
  \SSP^\ell_{n+1}(t) &= \{ \langle w, i, j\rangle \in \branch_t \mid
  \exists w_1 \in \SP^\ell_n(t|_{wi}) \colon \abs{\branch_t(wi, w_1)
    \cap \SP^\ell_n(t)} \geq \ell^{n+1}, \\*
  & \phantom{{}= \{ \langle w, i, j\rangle \in \branch_t \mid {}}
  \exists w_2 \in \SP^\ell_n(t|_{wj}) \colon \abs{\branch_t(wj, w_2)
    \cap \SP^\ell_n(t)} \geq \ell^{n+1} \} \\*
  \SP^\ell_{n+1}(t) &= \{ w \mid \langle w, i, j\rangle \in
  \SSP^\ell_{n+1}(t) \} \enspace.
\end{align*}
Intuitively, all branching positions have level~$0$ (for any
distance~$\ell$) and a branching position~$w$ has level~$n+1$ if there
are 2~paths in different direct subtrees below~$w$ that both have at
least $\ell^{n+1}$~branching positions of level~$n$ along the path.
Clearly, $\SSP^\ell_{n+1}(t) \subseteq \SSP^\ell_n(t)$ and
$\SP^\ell_{n+1}(t) \subseteq \SP^\ell_n(t)$ for all $n \in \nat$ and
$t \in T_\Sigma$.

\begin{figure}[t]
  \centering
  \includegraphics{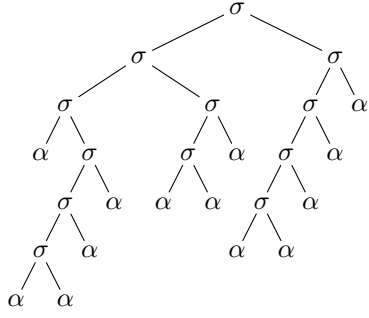}
  \caption{Tree used in Example~\protect{\ref{ex:sp}}.}
  \label{fig:sp}
\end{figure}

\begin{example}
  \label{ex:sp}
  Let $t$~be the tree depicted in Figure~\ref{fig:sp}.  Then 
  \begin{align*}
    \SP^2_0(t) &= \{\varepsilon, 1, 11, 112, 1121, 11211, 12, 121, 2,
    21, 211, 2111 \} = \branch_t \\
    \SP^2_1(t) &= \{\varepsilon, 1\} \\
    \SP^2_2(t) &= \emptyset \enspace.
  \end{align*}
\end{example}

\begin{lemma} \upshape
  \label{lm:SP1}
  Let $\varphi \colon T_\Gamma \to T_\Delta$~be a complete
  tree homomorphism.  Moreover, let $t \in T_\Gamma$ and $\ell, n \in
  \nat$.  If $\langle \varepsilon, i, j\rangle \in \SSP^\ell_n(t)$,
  then there exists $\langle v, i', j'\rangle \in
  \SSP^\ell_n(\varphi(t))$ such that
  \begin{compactitem}
  \item $v \in \pos(\varphi(t(\varepsilon)))$,
  \item $vi' \preceq v_1$ for some $v_1 \in
    \pos_{x_i}(\varphi(t(\varepsilon)))$, and
  \item $vj' \preceq v_2$ for some $v_2 \in
  \pos_{x_j}(\varphi(t(\varepsilon)))$.
  \end{compactitem}
\end{lemma}

\begin{proof}
  Let $t = \gamma(\seq t1k)$ with $\gamma \in \Gamma_k$ and $\seq t1k
  \in T_\Gamma$, and let $u = \varphi(t)$.  We prove the statement by
  induction on~$n$.  In the induction base, we have $n = 0$ and
  $\langle \varepsilon, i, j\rangle \in \SSP^\ell_0(t)$.  We
  arbitrarily select $v_1 \in \pos_{x_i}(\varphi(\gamma))$ and $v_2
  \in \pos_{x_j}(\varphi(\gamma))$, which are positions of the
  variables $x_i$~and~$x_j$ in~$\varphi(\gamma)$, respectively.  Since
  $\varphi$~is complete, such positions exist.  Let $v = \lca(v_1,
  v_2)$ be the longest common prefix, and let $v_1 = vi'v'_1$ and $v_2
  = vj'v'_2$.  Clearly, $v \in \pos(\varphi(\gamma))$ and $\langle v,
  i', j'\rangle \in \branch_{\varphi(t)}$ because $v_1$~and~$v_2$ are in the
  distinct direct subtrees $i'$~and~$j'$ below~$v$.  Trivially, $vi'
  \preceq v_1$ and $vj' \preceq v_2$, which completes the induction
  base. 
 
  In the induction step, let $\langle \varepsilon, i, j\rangle \in
  \SSP^\ell_{n+1}(t)$ and $w_1 \in \SP^\ell_n(t_i)$ and $w_2 \in
  \SP^\ell_n(t_j)$ be the required positions of level~$n$ such that
  \[ \abs{\branch_t(i, w_1) \cap \SP^\ell_n(t)} \geq \ell^{n+1} \qquad
  \text{and} \qquad \abs{\branch_t(j, w_2) \cap \SP^\ell_n(t)} \geq
  \ell^{n+1} \enspace. \]  We observe that
  \[ u = \varphi(t) = \varphi(\gamma)[\varphi(t_1), \dotsc,
  \varphi(t_k)] \enspace. \] Now, we follow a similar approach as in
  the induction base.  Since $\varphi$~is complete, the tree~$u$
  contains the subtrees $\varphi(t_i)$~and~$\varphi(t_j)$ at any of
  the positions $v_1 \in \pos_{x_i}(\varphi(\gamma))$ and $v_2 \in
  \pos_{x_j}(\varphi(\gamma))$, respectively.  As before, we let $v =
  \lca(v_1, v_2) \in \pos(\varphi(\gamma))$ be the longest common
  prefix, and let $v_1 = vi'v'_1$ and $v_2 = vj'v'_2$.  Clearly,
  $\langle v, i', j'\rangle \in \branch_u$ is a branching position.
  It remains to show that $\langle v, i', j'\rangle \in
  \SSP^\ell_{n+1}(u)$ because the remaining conditions of the lemma
  are already fulfilled.

  Due to the fact that $w_1 \in \SP^\ell_n(t_i)$, we have $\langle
  \varepsilon, i'', j''\rangle \in \SSP^\ell_n(t_i|_{w_1})$ for some
  $i'', j'' \in \nat$.  Moreover, since $\varphi$~is complete, the
  subtree~$\varphi(t_i|_{w_1})$ occurs in~$u|_{v_1}$.  By the
  induction hypothesis, there exist $v''_1 \in
  \SP^\ell_n(\varphi(t_i|_{w_1}))$.  Moreover, $v_1v'v''_1 \in
  \SP^\ell_n(u)$ for some $v' \in \pos(u|_{v_1})$ such that $u|_{v'} =
  \varphi(t_i|_{w_1})$.  We will only verify the condition
  \begin{equation}
    \label{eq}
    \abs{\branch_u(vi', v'_1v'v''_1) \cap \SP^\ell_n(u)} \geq \ell^{n+1}
  \end{equation}
  because the proof for the second path works analogously.  

  Let $w'_1 \in \branch_t(i, w_1) \cap \SP^\ell_n(t)$ be any position
  of level~$n$ along the path from~$i$ to~$iw_1$ such that $w'_1 \prec
  iw_1$.  Let $i''' \in \nat$ be the unique integer such that
  $w'_1i''' \preceq iw_1$.  Since $w'_1, iw_1 \in \SP^\ell_n(t)$ we
  can conclude that there exist $i'', j'' \in \nat$ such that
  $i'''\neq j''$ and that $\langle w'_1, i'', j''\rangle \in
  \SSP^\ell_n(t)$.  Then $\langle w'_1, i''', j'' \rangle \in
  \SSP^\ell_n(t)$ because $iw_1 \in \SP^\ell_n(t)$.  In other words,
  $\langle \varepsilon, i''', j''\rangle \in \SSP^\ell_n(t|_{w'_1})$.
  Since $\varphi$~is complete, the translation~$\varphi(t|_{w'_1})$
  occurs in~$\varphi(t_i)$.  By the induction hypothesis, there exists
  $\langle v''', i_1, j_1 \rangle \in
  \SSP^\ell_n(\varphi(t|_{w'_1}))$, which yields $\langle v''v''',
  i_1, j_1\rangle \in \SSP^\ell_n(\varphi(t_i))$ for some position
  $v'' \in \pos(\varphi(t_i))$ with $\varphi(t_i)|_{v''} =
  \varphi(t_{w'_1})$, such that $v'v''v'''i_1 \preceq v'_1$.  Clearly,
  \[ v_1v'v''v''' \in \branch_u(vi', v'_1v'v''_1) \cap \SP^\ell_n(u)
  \enspace. \] Moreover, the two conditions that $v''' \in
  \pos(\varphi(t(w'_1)))$ and that $\varphi$~is complete guarantee
  that for each selection of~$w'_1$ we obtain a different position
  $v_1v'v''v''' \in \branch_u(vi', v'_1v'v''_1) \cap
  \SP^\ell_n(u)$. This verifies (\ref{eq}), because there are at least
  $\ell^{n+1}$ possible selections of $w'_1$. \qed
\end{proof}

\begin{lemma} \upshape
  \label{lm:SP2}
  Let $\psi \colon T_\Gamma \to T_\Sigma$ be a linear tree
  homomorphism.  Moreover, let $t \in T_\Gamma$ and $\ell, n \in \nat$
  be such that $\ell > \height(\psi(\gamma'))$ for all symbols
  $\gamma' \in \Gamma$.  If $\langle v, i', j'\rangle \in
  \SSP^\ell_{n+1}(\psi(t))$ with $v \in \pos(\psi(t(\varepsilon)))$,
  then there exists $\langle \varepsilon, i, j\rangle \in
  \SSP^\ell_n(t)$ such that
  \begin{compactitem}
  \item $vi' \preceq v_1$ for some $v_1 \in
    \pos_{x_i}(\psi(t(\varepsilon)))$ and
  \item $vj' \preceq v_2$ for some $v_2 \in
  \pos_{x_j}(\psi(t(\varepsilon)))$.
  \end{compactitem}
\end{lemma}

\begin{proof}
  Let $t = \gamma(\seq t1k)$ with $\gamma \in \Gamma_k$ and $\seq t1k
  \in T_\Gamma$, and let $u = \psi(t)$ be its image.  Moreover, let
  $\langle v, i', j'\rangle \in \SSP^\ell_{n+1}(u)$ with $v \in
  \pos(\psi(\gamma))$.  By definition, there exist long paths in~$u$
  starting at $vi'$~and~$vj'$, which are longer than~$\ell$, which in
  turn is longer than any path in~$\psi(\gamma')$ for every $\gamma'
  \in \Gamma$.  Consequently, there exist $1 \leq i, j \leq k$ and
  \[ v_1 \in \pos_{x_i}(\psi(\gamma)) \qquad \text{and} \qquad v_2 \in
  \pos_{x_j}(\psi(\gamma)) \]
  such that $vi' \preceq v_1$ and $vj' \preceq v_2$.  Since $\psi$~is
  linear, we know that $i \neq j$.  It remains to prove that $\langle
  \varepsilon, i, j\rangle \in \SSP^\ell_n(t)$ as the additional
  properties of the lemma are already satisfied.  We prove this
  remaining property by induction on~$n$, and the induction base is
  already proven as $\langle \varepsilon, i, j\rangle \in \branch_t$.  

  We proceed with the induction step.  By definition, there exist
  positions $v'_1 \in \SP^\ell_n(u|_{vi'})$ and $v'_2 \in
  \SP^\ell_n(u|_{vj'})$ such that 
  \[ \abs{\branch_u(vi', v'_1) \cap \SP^\ell_n(u)} \geq \ell^{n+1}
  \qquad \text{and} \qquad \abs{\branch_u(vj', v'_2) \cap
    \SP^\ell_n(u)} \geq \ell^{n+1} \enspace. \]
  Let $w_1 \in \pos(t_i)$ be the position that creates the
  symbol~$u(vi'v'_1)$.    We only prove the required property for the
  path $\branch_t(i, w_1)$.

  Since $w_1$~creates the symbol~$u(vi'v'_1)$, there exists a
  position~$v'' \in \pos(u)$ such that $vi' \preceq v'' \preceq
  vi'v'_1$ and $u|_{v''} = \psi(t|_{w_1})$.  Moreover, by $v'_1 \in
  \SP^\ell_n(u|_{vi'})$ there exists $\langle v''_1, i'_1, j'_1\rangle
  \in \SSP^\ell_n(u|_{v''})$ with $vi'v'_1 = v''v''_1$, which yields
  that $v''_1 \in \pos(\psi(t(w_1)))$.  With the help of the induction
  hypothesis we conclude the existence of $\langle \varepsilon, i_1,
  j_1\rangle \in \SSP^\ell_n(t|_{w_1})$, and thus $w_1 \in
  \SP^\ell_n(t)$.  Finally, for every $v' \in \branch_u(vi', v'_1)
  \cap \SP^\ell_n(u)$.  Let $w' \in \pos(t_i)$ with $i \preceq w'
  \preceq w_1$ be the position that creates the symbol~$u(v')$, and
  let $v'' \in \pos(u)$ be the position such that $vi' \preceq v''
  \preceq v'$ and $u|_{v''} = \psi(t|_{w'})$.  Since $v' \in
  \SP^\ell_n(u)$ there exists $\langle v''_1, i'_1, j'_1\rangle \in
  \SSP^\ell_n(u|_{v''})$ with $v' = v''v''_1$, which yields that
  $v''_1 \in \pos(\psi(t(w')))$.  Moreover, let $i''_1 \in \nat$ be
  the unique integer such that $v'i''_1 \preceq vi'v'_1$.
  Consequently, $\langle v''_1, i''_1, j'_1\rangle \in
  \SSP^\ell_n(u|_{v''})$ because $vi'v'_1 \in \SP^\ell_n(u)$.  By the
  induction hypothesis, there exists $\langle \varepsilon, i_1,
  j_1\rangle \in \SSP^\ell_n(t|_{w'})$ such that
  \begin{compactitem}
  \item $v''_1i''_1 \preceq v'''_1$ for some $v'''_1 \in
    \pos_{x_{i_1}}(\psi(t(w')))$ and
  \item $v''_1j'_1 \preceq v'''_2$ for some $v'''_2 \in
  \pos_{x_{j_1}}(\psi(t(w')))$.
  \end{compactitem}
  Moreso, we have $w' \in \branch_t(i, w_1) \cap \SP^\ell_n(t)$, and
  $\psi(t|_{w'i_1}) = u|_{v'''}$ for some position $vi' \preceq v''' \preceq
  vi'v'_1$, which determines another position of~$\branch_u(vi',
  v'_1)$.  Since at most~$\ell$ positions of~$\branch_u(vi', v'_1)$
  can be created by a single symbol of~$t$ and $\abs{\branch_u(v_1,
    v'_1) \cap \SP^\ell_n(u)} \geq \ell^{n+1} - \ell + 1 = (\ell^n -
  1 + \ell^{-1}) \ell$.  Consequently, we obtain at least~$\ell^n$
  positions in $\branch_t(i, w_1) \cap \SP^\ell_n(t)$, which completes the
  induction step and the proof.  \qed
\end{proof}

Next, we combine the previous two lemmas into the main statement that
will be used to prove the infinity of several composition hierarchies.
In essence, we show that a bimorphism in $\text B(\text l, \text c)$
can reduce the level by at most~$1$ (while keeping the
distance~$\ell$).

\begin{theorem} \upshape
  \label{thm:SP}
  Let $B = (\psi, L, \varphi)$ be a bimorphism such that $\psi \colon
  T_\Gamma \to T_\Sigma$ is linear and $\varphi \colon T_\Gamma \to
  T_\Delta$ is complete.  Moreover, let $(t, u) \in B$, and let $\ell
  \in \nat$ be such that $\ell > \height(\psi(\gamma))$ for every
  $\gamma \in \Gamma$.  If $\SP^\ell_{n+1}(t) \neq \emptyset$, then
  $\SP^\ell_n(u) \neq \emptyset$.
\end{theorem}

\begin{proof}
  Since $(t, u) \in B$, there exists $s \in L$ such that $\psi(s) = t$
  and $\varphi(s) = u$.  By assumption, we have that
  $\SP^\ell_{n+1}(\psi(s)) \neq \emptyset$, which by
  Lemma~\ref{lm:SP2} yields that $\SP^\ell_n(s) \neq \emptyset$.
  Using the completeness of~$\varphi$ and Lemma~\ref{lm:SP1}, we
  obtain that $\SP^\ell_n(u) \neq \emptyset$ as desired.  \qed
\end{proof}

Now we can simply chain Theorem~\ref{thm:SP} to show that an $n$-fold
composition of tree transformations of~$\text B(\text l, \text c)$ can
decrease the level by at most~$n$ (for a suitable distance~$\ell$).

\begin{corollary}[of~\protect{Theorem~\ref{thm:SP}}] \upshape
  \label{cor:SP}
  Let $n \geq 1$, and for every $1 \leq i \leq n$, let $B_i =
  (\psi_i, L_i, \varphi_i)$ be a bimorphism such that $\psi_i$~is
  linear and $\varphi_i$ is complete.  Moreover, let $\varphi_i \colon
  T_{\Gamma_i} \to T_{\Delta_i}$ and $\psi_{i+1} \colon
  T_{\Gamma_{i+1}} \to T_{\Delta_i}$ for every $1 \leq i < n$.
  Finally, let $\ell \in \nat$ be such that $\ell >
  \height(\psi_i(\gamma))$ for every $1 \leq i \leq n$ and $\gamma \in
  \Gamma_i$, and let $(t, u) \in B_1 \mathbin; \dotsm \mathbin; B_n$.
  If $\SP^\ell_{n+1}(t) \neq \emptyset$, then $\SP^\ell_1(u) \neq
  \emptyset$.
\end{corollary}

\begin{figure}[t]
  \centering
  \includegraphics[scale=1]{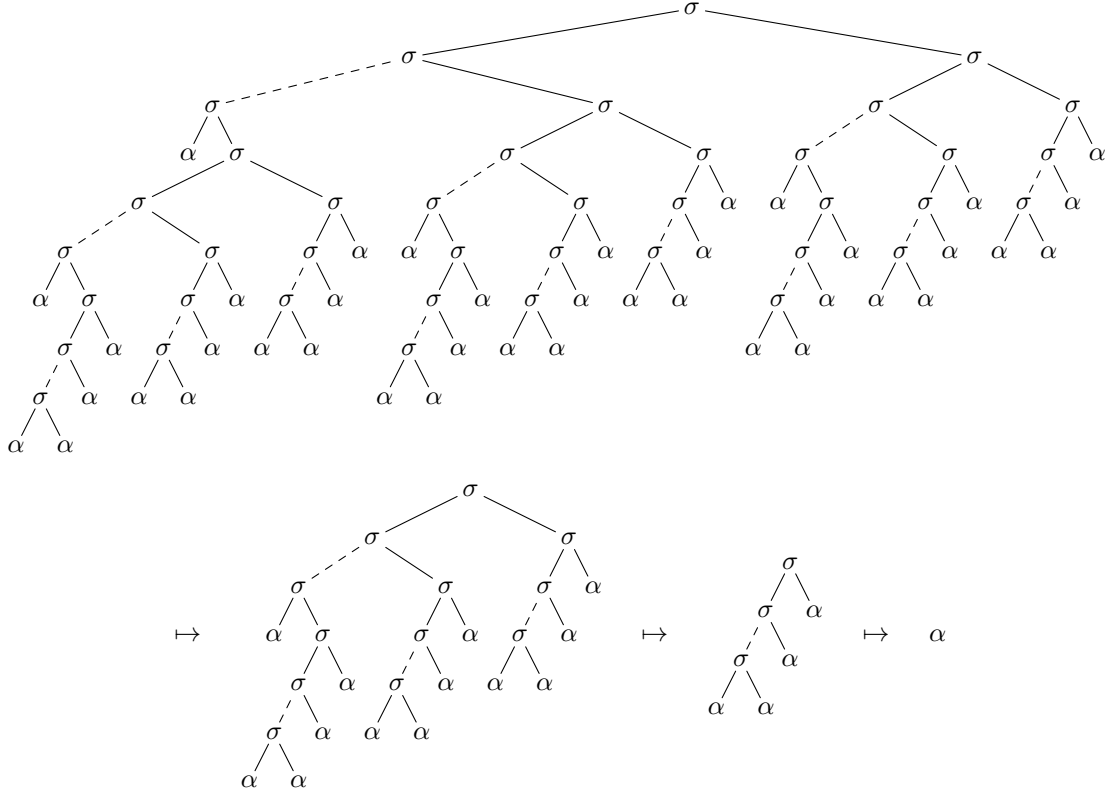}
  \caption{Illustration of the $3$-fold composition of the tree
    transformation computed by the
    $\mathord{\not{\!\varepsilon}}$nl-xt~$M$ of
    Example~\protect{\ref{ex:ns}}.}
  \label{fig:reduce}
\end{figure}

It remains to demonstrate a tree transformation that can be computed
by $(n+1)$~$\mathord{\not{\!\varepsilon}}$nl-xt and that reduces the
level of positions from~$n+1$ to~$0$.  Clearly, this tree
transformation cannot be computed by an $n$-fold composition of tree
transformations from~$\text B(\text l, \text c)$ because the output
tree should contain a position of level~$1$ by
Corollary~\ref{cor:SP}. We make sure that the assumptions of
Corollary~\ref{cor:SP} are satisfied.

\begin{example}
  \label{ex:ns}
  Let $M = (Q, \Sigma, \Sigma, \{\star\}, R)$ be the
  $\mathord{\not{\!\varepsilon}}$nl-xt such that 
  \begin{compactitem}
  \item $Q = \{\star, q\}$ and $\Sigma = \{\sigma^{(2)}, \alpha^{(0)}\}$, and
  \item $R$~contains exactly the following rules
    \begin{align*}
      \sigma(\star, \alpha) &\stackrel{\star,q} \longrightarrow \star &
      \sigma(\star, q) &\stackrel{\star,q} \longrightarrow
      \sigma(\star, q) & \alpha &\stackrel \star \longrightarrow 
      \alpha
    \end{align*}
  \end{compactitem}
  The $3$-fold composition $M \mathbin; M \mathbin; M$ can compute the
  transformation indicated in Figure~\ref{fig:reduce}.
\end{example}

We use the tree transformation computed by the
$\mathord{\not{\!\varepsilon}}$nl-xt~$M$ of Example~\ref{ex:ns}, and
show that $n$~transformations from~$\text B(\text l, \text c)$
cannot compute the tree transformation~$M^{n+1}$.

\begin{lemma} \upshape
  \label{thm:ns}
  For every $n \geq 1$,
   \[ \mathord{\not{\!\varepsilon}}\text{nl-XT}^{n+1} \not\subseteq
   \text{B(l,c)}^n \enspace. \]
\end{lemma}

\begin{proof}
  Let $L_0 = \{\alpha\}$.  For every $i \geq 1$, let $C_i$~and~$L_i$
  be the smallest sets such that
  \begin{compactitem}
  \item $\sigma(x_1, t) \in C_i$ for every $t \in L_{i-1}$ and $c[c']
    \in C_i$ for all $c, c' \in C_i$, and
  \item $L_i = \{ c[\alpha] \mid c \in C_i\}$.
  \end{compactitem}
  It is easy to see that $n+1$~compositions of the tree
  transformation~$M$ computed by the
  $\mathord{\not{\!\varepsilon}}$nl-xt~$M$ of Example~\ref{ex:ns} can
  relate $L_{n+2}$~and~$L_1$ (see Figure~\ref{fig:reduce}); more
  precisely, for every $t \in L_{n+2}$ there exists $u \in L_1$ such
  that $(t, u) \in M^{n+1}$.  With the help of Corollary~\ref{cor:SP}
  we can complete the proof if $\SP^\ell_{n+1}(t) \neq \emptyset$ for
  some $t \in L_{n+2}$ and suitably large~$\ell \in \nat$.

  Let $\ell \in \nat$ be fixed.  Thus, we now prove that for every $n
  \in \nat$ there exists a tree~$t \in L_{n+1}$ such that
  $\SP^\ell_n(t) \neq \emptyset$ by induction on~$n$.  In fact, we
  prove the stronger statement that there exists $t \in L_{n+1}$ and
  $w \in \pos(t)$ such that $\abs{\branch_t(\varepsilon, w) \cap
    \SP^\ell_n(t)} \geq \ell^{n+1}$.  For $n = 0$, we select the
  tree~$c^\ell[\alpha] \in L_1$, where $c = \sigma(x_1, \alpha)$, and
  the position $w = 1^{\ell-1}$.  Since $\SP^\ell_0(t) =
  \branch_t(\varepsilon, w)$, the selection of $t$~and~$w$ fulfills
  the requirements.  In the induction step, there exist a tree~$t \in
  L_{n+1}$ and $w \in \pos(t)$ such that $\abs{\branch_t(\varepsilon,
    w) \cap \SP^\ell_n(t)} \geq \ell^{n+1}$.  We consider the tree~$t'
  = c^{\ell^{n+2}}[c[\alpha]]$ with $c = \sigma(x_1, t)$ and the
  position~$w' = 1^{\ell^{n+2}-1}$.  Obviously, $t' \in L_{n+2}$ and
  $w'' \in \SP^\ell_{n+1}(t')$ for every $w'' \preceq w'$ (via the
  paths $12w$~and~$2w$), which completes our induction and proof.
  \qed
\end{proof}

Now we are able to prove that the composition hierarchy of
$\mathord{\not{\!\varepsilon}}$nl-XT and several other classes is
infinite.

\begin{theorem} \upshape
  \label{cor:eps}
  For every $n \geq 1$,
  \begin{align*}
    \text{B(l,c)}^{n} &\subsetneq \text{B(l,c)}^{n+1} &
    \text{[s][n]l-XT}^n &\subsetneq \text{[s][n]l-XT}^{n+1} \\
    \mathord{\not{\!\varepsilon}}\text{nl-XT}^n & \subsetneq
    \mathord{\not{\!\varepsilon}}\text{nl-XT}^{n+1} &
    (\text{[s][n]l-XT}^{\text R})^n &\subsetneq
    (\text{[s][n]l-XT}^{\text R})^{n+1} \enspace.
  \end{align*}
\end{theorem}

\begin{proof}
  First, we note that all inclusions are trivial.  Thus, we only need
  to argue the strictness.  By \cite[Theorem~17]{mal07e} we have
  $\text B(\text{lcs}, \text{lc}) =
  \text{$\mathord{\not{\!\varepsilon}}$nl-XT}$, hence
  \begin{align*}
    \mathord{\not{\!\varepsilon}}\text{nl-XT}^n \subseteq
    \text{B(l,c)}^n  \qquad \text{and} \qquad
    \mathord{\not{\!\varepsilon}}\text{nl-XT}^{n+1} \subseteq
    \text{B(l,c)}^{n+1} \enspace.
  \end{align*}
  These facts together with Lemma~\ref{thm:ns} imply the strictness of
  the two inclusions on the left.  Moreover,
  \begin{align*}
    \cdots {} \subseteq (\text{[s][n]l-XT}^{\text R})^n &\subseteq
    \text{[s][n]l-XT}^{n+1} \subseteq (\text{[s][n]l-XT}^{\text
      R})^{n+1} \subseteq \text{[s][n]l-XT}^{n+2} \subseteq
    {} \cdots  
  \end{align*}
  for all $n \geq 1$, which shows that the composition hierarchy
  of~$\text{[s][n]l-XT}$ is infinite if and only if the composition
  hierarchy of~$\text{[s][n]l-XT}^{\text R}$ is infinite.  Thus, it
  remains to show the latter property.  

  Using simple symmetry, we observe that $\text{snl-XT} =
  \mathord{\not{\!\varepsilon}}\text{nl-XT}^{-1}$, which together with
  the symmetric version of Lemma~\ref{thm:ns} yields
  $\text{snl-XT}^{n+1} \not\subseteq \text{B(c,l)}^n$.  Furthermore,
  $\text B(\text{lc}, \text l) = \text{l-XT}^{\text R}$ by
  \cite[Theorem~4]{mal07e}, so naturally,
  \[ (\text{[s][n]l-XT}^{\text R})^n \subseteq \text B(\text c, \text l)^n
  \qquad \text{and} \qquad \text{snl-XT}^{n+1} \subseteq
  (\text{[s][n]l-XT}^{\text R})^{n+1} \]
  for all $n \geq 1$.  Taken together with $\text{snl-XT}^{n+1}
  \not\subseteq \text{B(c,l)}^n$ we obtain the desired
  strictness. \qed
\end{proof}

\begin{table}[t]
  \begin{center}
    \begin{tabular}{|l|c|}
      \hline
      &  \\[-1ex]
      \; Class with infinite composition hierarchy \; & Proved in
      \\[1ex] \hline
      &  \\[-1ex]
      \; $\text{B(l,c)}$, $\mathord{\not{\!\varepsilon}}\text{ln-XT}$, 
      $\text{[s][n]l-XT}$, $\text{[s][n]l-XT}^{\text R}$ \; &
      Theorem~\ref{cor:eps} \\[1ex] \hline
      &  \\[-1ex]
      \; $\text T$ & \cite[Theorem~3.14]{eng82b} \\[1ex]
      \; $\text{$\mathord{\not{\!\varepsilon}}$-XT}$ &  \;
      $\text{$\mathord{\not{\!\varepsilon}}$-XT} \subseteq \text T^2$
      and $(\text T^n \mid n \geq 1)$ is infinite \; \\[1ex]
      \; B(c,c) &  \cite{arndau75} and
      \cite[Section~II-2-2-3-4]{dau77} \\[1ex] \hline
    \end{tabular}
  \end{center}
  \caption{Summary of the results of
    Section~\protect{\ref{sec:complete}}.}
  \label{tab:inf}
\end{table}

Table~\ref{tab:inf} summarizes our results of this section.  For the
sake of completeness, we mention some additional results from the
literature, where $\text T$~stands for the class of all tree
transformations computable by top-down tree transducers~\cite{eng75},
and $\text{$\mathord{\not{\!\varepsilon}}$-XT}$~stands for the class
of tree transformations computable by $\varepsilon$-free extended
top-down tree transducers~\cite{malgrahopkni07}.  The mentioned result
$\text{$\mathord{\not{\!\varepsilon}}$-XT} \subseteq \text T^2$ can be
concluded from~\cite[Theorem~4.8]{malgrahopkni07}.

\bibliography{extra}

\begin{thebibliography}{10}
\providecommand{\url}[1]{\texttt{#1}}
\providecommand{\urlprefix}{URL }

\bibitem{arndau75}
Arnold, A., Dauchet, M.: Transductions inversibles de for\^ets. Th\`ese 3\`eme
  cycle {M.~Dauchet}, Universit\'e de Lille (1975)

\bibitem{arndau76}
Arnold, A., Dauchet, M.: Bi-transductions de for{\^e}ts. In: ICALP. pp. 74--86.
  Edinburgh University Press (1976)

\bibitem{arndau82}
Arnold, A., Dauchet, M.: Morphismes et bimorphismes d'arbres. Theoret.\
  Comput.\ Sci.  20(1),  33--93 (1982)

\bibitem{bak79}
Baker, B.S.: Composition of top-down and bottom-up tree transductions. Inform.\
  and Control  41(2),  186--213 (1979)

\bibitem{chi06}
Chiang, D.: An introduction to synchronous grammars. In: ACL. Association for
  Computational Linguistics (2006), part of a tutorial given with K.~Knight

\bibitem{dau77}
Dauchet, M.: Transductions de for\^ets bimorphismes de magmo\"\i des.
  Premi\`ere th\`ese, Universit\'e de Lille (1977)

\bibitem{eng75}
Engelfriet, J.: Bottom-up and top-down tree transformations---a comparison.
  Math.\ Systems Theory  9(3),  198--231 (1975)

\bibitem{eng77}
Engelfriet, J.: Top-down tree transducers with regular look-ahead. Math.\
  Systems Theory  10(1),  289--303 (1977)

\bibitem{eng82b}
Engelfriet, J.: Three hierarchies of transducers. Math.\ Systems Theory  15(2),
   95--125 (1982)

\bibitem{engman03b}
Engelfriet, J., Maneth, S.: Macro tree translations of linear size increase are
  {MSO} definable. SIAM J.\ Comput.  32(4),  950--1006 (2003)

\bibitem{fulmalvog10}
F\"ul\"op, Z., Maletti, A., Vogler, H.: Preservation of recognizability for
  synchronous tree substitution grammars. In: ATANLP. pp. 1--9. Association for
  Computational Linguistics (2010)

\bibitem{fulmalvog11}
F\"ul\"op, Z., Maletti, A., Vogler, H.: Weighted extended tree transducers.
  Fundam.\ Inform.  111(2),  163--202 (2011)

\bibitem{fulvog98}
F{\"u}l{\"o}p, Z., Vogler, H.: Syntax-Directed Semantics---Formal Models Based
  on Tree Transducers. EATCS Monographs on Theoret.\ Comput.\ Sci., Springer
  (1998)

\bibitem{gecste84}
G{\'e}cseg, F., Steinby, M.: Tree Automata. Akad{\'e}miai Kiad{\'o}, Budapest
  (1984)

\bibitem{gecste97}
G{\'e}cseg, F., Steinby, M.: Tree languages. In: Rozenberg, G., Salomaa, A.
  (eds.) Handbook of Formal Languages, vol.~3, chap.~1, pp. 1--68. Springer
  (1997)

\bibitem{malgrahopkni07}
Graehl, J., Hopkins, M., Knight, K., Maletti, A.: The power of extended
  top-down tree transducers. SIAM J.\ Comput.  39(2),  410--430 (2009)

\bibitem{graknimay08}
Graehl, J., Knight, K., May, J.: Training tree transducers. Comput.\ Linguist.
  34(3),  391--427 (2008)

\bibitem{knigra05}
Knight, K., Graehl, J.: An over\-view of probabilistic tree transducers for
  natural language processing. In: CICLing. {\upshape LNCS}, vol. 3406, pp.
  1--24. Springer (2005)

\bibitem{mal07e}
Maletti, A.: Compositions of extended top-down tree transducers. Inform.\ and
  Comput.  206(9--10),  1187--1196 (2008)

\bibitem{mal11d}
Maletti, A.: Tree transformations and dependencies. In: MOL. {LNAI}, vol. 6878,
  pp. 1--20. Springer (2011)

\bibitem{mayknivog10}
May, J., Knight, K., Vogler, H.: Efficient inference through cascades of
  weighted tree transducers. In: ACL. pp. 1058--1066. Association for
  Computational Linguistics (2010)

\bibitem{rou70}
Rounds, W.C.: Mappings and grammars on trees. Math.\ Systems Theory  4(3),
  257--287 (1970)

\bibitem{tha70}
Thatcher, J.W.: Generalized$^2$ sequential machine maps. J. Comput.\ System
  Sci.  4(4),  339--367 (1970)

\end{thebibliography}
\bibliographystyle{splncs03}

\end{document}